%% file: main.tex
\documentclass[11pt]{article}
\usepackage[a4paper,margin=2.3cm]{geometry}
\usepackage[T1]{fontenc}
\usepackage{amsmath,amsthm,mathtools}
\usepackage{newtxtext,newtxmath}
\usepackage[scaled=0.92]{helvet}
\DeclareMathVersion{sfdigits}
\SetSymbolFont{operators}{sfdigits}{OT1}{phv}{m}{n}
\SetMathAlphabet{\mathbf}{sfdigits}{OT1}{phv}{b}{n}
\usepackage{microtype}
\usepackage{aliascnt}
\usepackage{physics}
\usepackage{graphicx,booktabs,multirow,array}
\usepackage{algorithm,algpseudocode}
\usepackage{placeins}
\usepackage{needspace}
\usepackage{tikz,xcolor}
\usepackage{xparse,xspace,suffix,enumitem}
\usepackage[numbers,super,sort&compress]{natbib}
\usepackage{caption}
\DeclareCaptionLabelSeparator{pipe}{\ $|$\ }
\DeclareCaptionFont{sfmath}{\sffamily\mathversion{sfdigits}}
\usepackage{titlesec}
\usepackage{chapterbib}
\usepackage{hyperref}
\usepackage[capitalize,nameinlink]{cleveref}
\RemoveFromHook{label}[firstaid/cleveref]
\hypersetup{hidelinks,pdfauthor={Bin Cheng, Ziyuan Wang, Ruixuan Deng, Jianxin Chen, Zhengfeng Ji}}
\allowdisplaybreaks
\titleformat{\section}{\sffamily\bfseries\large}{\thesection}{.7em}{}
\titleformat{\subsection}{\sffamily\bfseries\normalsize}{\thesubsection}{.7em}{}
\titleformat{\subsubsection}{\sffamily\bfseries\normalsize}{\thesubsubsection}{.7em}{}
\titlespacing*{\section}{0pt}{1.3em plus .3em minus .2em}{.4em}
\titlespacing*{\subsection}{0pt}{1em plus .2em minus .2em}{.25em}
\titlespacing*{\subsubsection}{0pt}{.8em plus .2em minus .2em}{.2em}
\renewcommand{\bibnumfmt}[1]{#1.}

\renewcommand{\figurename}{Fig.}
\newcommand{\manuscripttitle}{Breaking the Treewidth Barrier in Quantum Circuit Simulation with Decision Diagrams}
\hypersetup{pdftitle={\manuscripttitle}}
\makeatletter
\renewcommand{\maketitle}{\begingroup\raggedright\sffamily
  {\fontsize{20}{24}\selectfont\bfseries\@title\par}\vspace{10pt}
  {\fontsize{11}{14}\selectfont\@author\par}\vspace{5pt}
  {\fontsize{8.5}{11}\selectfont\affiliations\par}\endgroup\vspace{12pt}}
\renewenvironment{abstract}{\par\begingroup\noindent\sffamily\mathversion{sfdigits}\small\ignorespaces}{\par\endgroup}
\makeatother

\theoremstyle{plain}
\newtheorem{theorem}{Theorem}
\newtheorem{lemma}[theorem]{Lemma}

\newtheorem{stheorem}{Theorem}
\newtheorem{slemma}[stheorem]{Lemma}
\newtheorem{sproposition}[stheorem]{Proposition}
\theoremstyle{definition}

\newtheorem{sdefinition}[stheorem]{Definition}
\theoremstyle{remark}

\crefname{stheorem}{Theorem}{Theorems}
\crefname{slemma}{Lemma}{Lemmas}
\crefname{sproposition}{Proposition}{Propositions}
\crefname{sdefinition}{Definition}{Definitions}
\crefname{proposition}{Proposition}{Propositions}
\crefname{algorithm}{Algorithm}{Algorithms}
\crefname{figure}{Fig.}{Figs.}
\Crefname{figure}{Fig.}{Figs.}

\newcommand{\calG}{\mathcal{G}}
\newcommand{\calR}{\mathcal{R}}
\newcommand{\calT}{\mathcal{T}}
\newcommand{\calU}{\mathcal{U}}
\newcommand{\calZ}{\mathcal{Z}}

\newcommand{\ZZ}{\mathbb{Z}}
\newcommand{\FF}{\mathbb{F}}
\newcommand{\bin}{\{0,1\}}
\newcommand{\ii}{\mathrm{i}}
\newcommand{\lup}{\mathsf{up}}
\newcommand{\size}[1]{\mathrm{B}(#1)}
\newcommand{\class}[1]{\textup{#1}\xspace}
\newcommand{\NP}{\class{NP}}
\newcommand{\Had}{\mathrm{H}}
\newcommand{\T}{\mathrm{T}}
\newcommand{\X}{\mathrm{X}}
\newcommand{\Y}{\mathrm{Y}}
\newcommand{\Z}{\mathrm{Z}}
\newcommand{\W}{\mathrm{W}}
\newcommand{\CZ}{\mathrm{CZ}}
\newcommand{\CCZ}{\mathrm{CCZ}}
\newcommand{\iSWAP}{\mathrm{iSWAP}}
\newcommand{\fSim}{\mathrm{fSim}}
\DeclareMathOperator{\lrw}{lrw}
\DeclareMathOperator{\rw}{rw}
\DeclareMathOperator{\pw}{pw}
\DeclareMathOperator{\tw}{tw}
\DeclareMathOperator{\cc}{cc}
\DeclareMathOperator{\poly}{poly}
\DeclareMathOperator{\locate}{\texttt{locate}}
\DeclareMathOperator{\intern}{\texttt{intern}}

\definecolor{ChannelColor}{rgb}{0.25,0.45,0.75}
\definecolor{ChannelTextColor}{rgb}{1,1,1}
\definecolor{LineColor}{rgb}{0,0,0}
\usetikzlibrary{positioning,calc,decorations.pathreplacing,shapes.misc,arrows.meta}
\tikzset{
  draw=LineColor, text=LineColor,
  control/.style={circle,fill=LineColor,minimum size=3,inner sep=0},
  gate/.style={draw,minimum size=14,fill=ChannelColor,inner sep=2,text=ChannelTextColor},
  target/.style={circle,draw,minimum size=6,inner sep=0},
  cross/.style={cross out,draw,minimum size=2.5,inner sep=0}
}

\title{\manuscripttitle}
\author{Bin Cheng\textsuperscript{1,\textdagger}, Ziyuan Wang\textsuperscript{2,3,\textdagger}, Ruixuan Deng\textsuperscript{2}, Jianxin Chen\textsuperscript{2,*} and Zhengfeng Ji\textsuperscript{2,3,*}}
\newcommand{\affiliations}{%
\textsuperscript{1}Centre for Quantum Technologies, National University of Singapore, Singapore\\
\textsuperscript{2}Department of Computer Science and Technology, Tsinghua University, Beijing, China\\
\textsuperscript{3}Zhongguancun Laboratory, Beijing, China\\
\textsuperscript{\textdagger}These authors contributed equally.\\
\textsuperscript{*}Corresponding authors: \href{mailto:chenjianxin@tsinghua.edu.cn}{chenjianxin@tsinghua.edu.cn}; \href{mailto:jizhengfeng@tsinghua.edu.cn}{jizhengfeng@tsinghua.edu.cn}}
\date{}

\begin{document}
\pdfbookmark[0]{\manuscripttitle}{maintext}
\maketitle

\begin{abstract}
  Classical simulation of quantum circuits is a critical tool for validating
  quantum hardware and probing the boundary between classical and quantum
  computational power.
  Existing state-of-the-art methods, notably tensor network approaches, have
  computational costs governed by the treewidth of the underlying circuit graph,
  making circuits with large treewidth intractable.
  This work rigorously analyzes and extends FeynmanDD, a decision diagram-based simulation
  method proposed in CAV 2025 by a subset of the authors.
  FeynmanDD evaluates sum-of-powers (SOP) expressions for circuit quantities such as output amplitudes and marginal probabilities using decision diagrams as the computational engine.
  We first show that given a variable ordering of width $w$, the complexity for constructing the diagram and evaluating the corresponding SOP for circuits over $\{\Had,\T,\CZ\}$ is exponential in the \emph{linear rank-width} of the circuit graph.
  As linear rank-width can be substantially smaller than treewidth and is at
  most larger than the treewidth by a logarithmic factor, our analysis
  demonstrates that for certain circuit families, FeynmanDD runs in polynomial time while contracting the corresponding tensor network requires exponential time and space.
  Moreover, we extend the original FeynmanDD method with weighted model counting, reducing SOP evaluation to two weighted counts on a single binary decision diagram.
  This takes the method beyond its original discrete gate sets to circuits with arbitrary single-qubit gates, with at most a threefold increase in linear rank-width.
  Given the prevalence of weighted model counting problems, from partition-function evaluation in many-body physics to probabilistic inference in artificial intelligence, this suggests that FeynmanDD has potential applications well beyond quantum circuit simulation.
\end{abstract}

\medskip

\section*{Introduction}\label{sec:introduction}

Classical simulation of quantum circuits plays a central role in advancing our
understanding of quantum computation.
On one hand, the limitation of classical methods for simulating universal
quantum circuits that are revealed during the study reinforces the belief that
quantum computers possess computational capabilities beyond those of classical
systems~\cite{Got98,Val02}.
On the other hand, classical simulation techniques provide indispensable tools
for validating experimental claims of quantum computation (see
e.g.~\citenum{KLR+08,MGE11,AAB+19}), particularly in the context of recent quantum
supremacy demonstrations~\cite{AAB+19,ZWD+20,WBC+21,MLA+22,GFZ+25}.
Identifying the precise point at which quantum computation outperform classical
ones is therefore a question of both practical and theoretical significance.

Over the past decades, a rich set of classical simulation techniques has
emerged.
Prominent among these are the efficient simulation of Clifford circuits enabled
by the celebrated Gottesman--Knill theorem~\cite{Got98}, algorithms exploiting
low T-count structures~\cite{BG16}, matchgate simulation
techniques~\cite{Val02}, and, most notably, tensor network-based
approaches~\cite{MS08,HZN+21,PVWC07,OGo19}.
Among these, tensor networks have emerged as particularly powerful tools.
They are widely used in many-body physics as the numerical computation engine to
study many-body phenomena~\cite{Vid04,Vid08,Oru14} and have become the method of
choice for simulating general quantum circuits, especially those lacking special
algebraic structures like in Clifford or matchgate circuits.

The performance of tensor network-based algorithms is governed by an important
graph parameter of the underlying circuit: the treewidth of the associated
tensor network~\cite{MS08}.
The treewidth measures how close a graph is to a tree and is, roughly speaking,
a measure of the complexity of entangling structure of the circuit.
Both the time and space complexity of tensor network contraction
are bounded by $2^{\order{w}} \poly(m)$, where $w$ is the treewidth and $m$ is the
number of gates in the circuit.
There are refinements and many variants of classical simulation methods that can
be characterized by other graph parameters, such as the path-width~\cite{ALM07},
the vertex congestion and modified cut width~\cite{OGo19}, but these are within
linearly or polylogarithmically factors of the treewidth and thus are only of
practical importance for implementation and do not affect the theoretical bounds
of the related algorithms.
From a parameterized complexity perspective~\cite{FG06}, these tensor
network-based methods are classical fixed-parameter tractable (FPT) algorithms
using treewidth as the parameter.

For generic quantum circuits with little inherent structure, tensor network
algorithms with complexity exponential in the treewidth have represented the
state of the art for decades.
This naturally raises the question:
\textit{Is the treewidth of the tensor network an inherent barrier to the
classical simulation of general quantum circuits, or can it be overcome
using fundamentally different simulation techniques?}

In this work, we answer this question in the affirmative.
We show that the recently introduced FeynmanDD method~\cite{WCYJ25} admits a
rigorous complexity analysis in terms of the so called \emph{linear rank-width}
of the same circuit graph.
The approach departs entirely from the tensor network paradigm, relying instead
on \emph{binary decision diagrams} (BDDs)~\cite{Bry86,Bry95,Weg00,Knu09} as the
core computational data structure.
Importantly, there exist families of circuits where the linear rank-width is
significantly smaller than the treewidth, suggesting the potential for
substantial speedups in classical simulation.
Thus, the FeynmanDD method opens a new regime for classical simulation beyond
the treewidth barrier, paving the way for its use in validating, benchmarking,
and advancing our understanding of quantum computation.

At the core of the FeynmanDD algorithm lies a Feynman path integral type of
simulation powered by an underlying decision diagram data structure.
For a circuit $C$ using many supported discrete gate sets, including $\mathcal{T} = \{\Had, \T, \CZ\}$, and the
$\iSWAP$ Google supremacy gate set
$\mathcal{G} = \{ \sqrt{\X}, \sqrt{\Y}, \sqrt{\W}, \iSWAP \}$\footnote{In Ref.~\citenum{WCYJ25},
  gate sets containing $\fSim$ were also supported.
  However, it is technically complicated to handle using the methods in this
  paper, so we do not consider it here.}, many simulation tasks such as
computing the amplitude, estimating the acceptance probability, and sampling the
outcome of a computational basis measurement can be reduced to the
computation of an exponential sum known as a sum-of-powers (SOP) form:
$\frac{1}{R}\sum_{\vb{x}} \omega_{r}^{f(\vb{x})}$, where $r$ is an integer modulus
determined by the gate set and is a fixed small constant, $R$ is a normalization
factor, and $f$ is a low-degree polynomial over $\{0,1\}$ taking integer
values modulo $r$.
The computation of the exponential sum in the SOP can be reduced to a constant
number of counting problems that compute the size of $\{ x \mid f(x) = j\}$.
The key insight of FeynmanDD is that if the function $f(x)$ can be represented by a succinct
$r$-terminal decision diagram, the counting can be performed efficiently.
We extend the original FeynmanDD method with weighted model counting (WMC).
For the SOPs considered here, we incorporate the phases from diagonal single-qubit gates into the weights and evaluate two weighted counts on a BDD for the quadratic part.
The simulation workflow is summarized in \cref{fig:feynmandd-overview}.
Even though the method borrows a lot from the classical decision diagram
literature, the heavy use of the efficient counting algorithm of BDDs to do
circuit analysis is, to the best of our knowledge, unique to the FeynmanDD
method.

\begin{figure}[t]
  \centering
  \includegraphics[width=\textwidth]{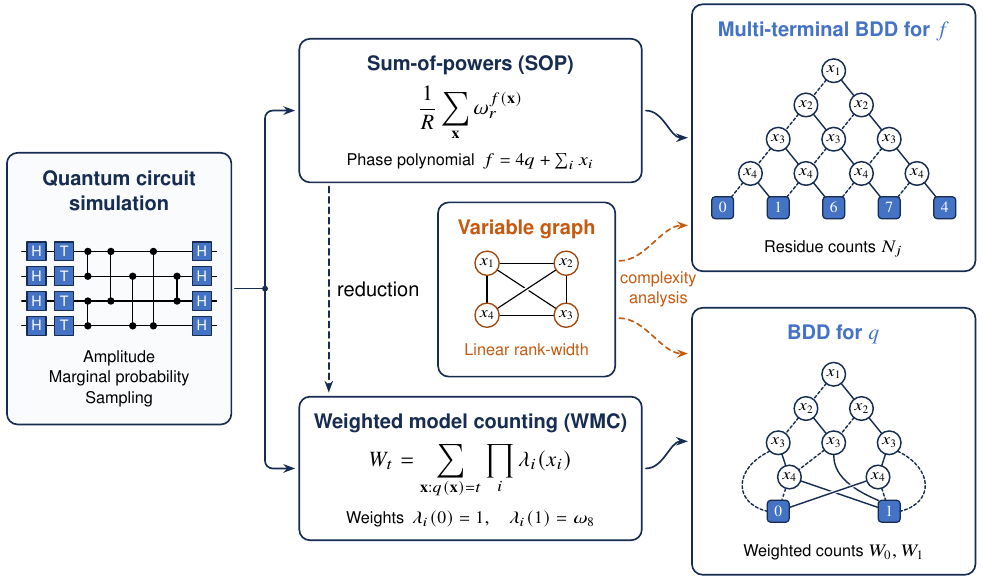}
  \caption{Schematic overview of quantum circuit simulation with FeynmanDD.
    Simulation tasks are reduced to SOP evaluation (upper branch) or weighted model counting (lower branch).
    The MTBDD for $f$ gives the residue counts $N_j$.
    The reduction from SOP evaluation to WMC gives the weighted counts $W_0$ and $W_1$ on a BDD for the binary quadratic part $q$.
    The dashed orange arrows indicate that the linear rank-width of the shared variable graph bounds the sizes of both diagrams.
    The circuit, graph, and two decision diagrams illustrate the same zero-to-zero amplitude $\sum_{j=0}^7 N_j\omega_8^j/16=(W_0-W_1)/16$ for the discrete gate set $\calT$.
    Here, $r=8$ and $q(\vb{x})=\sum_{i<j}x_ix_j$ over $\FF_2$.
    Replacing a $\T$ gate by a diagonal gate with entries $1,e^{\ii\theta}$ changes its weight $\lambda_i(1)$ from $\omega_8$ to $e^{\ii\theta}$ and leaves the BDD for $q$ unchanged.
    Dashed and solid DD edges denote assignments $0$ and $1$, respectively.
  }\label{fig:feynmandd-overview}
\end{figure}

The decision diagrams employed in FeynmanDD have long been a standard tool in
classical computer science for representing and manipulating Boolean circuits
and functions efficiently.
Introduced in the 1980s~\cite{Bry86}, BDDs have since become central to
applications such as circuit synthesis, formal verification, and model checking.
Their impact has been widely recognized, and many textbooks are dedicated to
this topic~\cite{Weg00,Knu09,MMBS04}.
In \emph{The Art of Computer Programming}, Knuth dedicated an entire section to
the theory and applications of BDDs~\cite{Knu09}.

BDDs possess several important features and properties, and the key result we
rely on to prove the linear rank-width bound is an explicit bound on the number
of nodes at any level of the diagram---that is, the number of distinct functions
obtained by assigning values to the preceding variables.
Specifically, for a function $f(x_{1}, x_{2}, \ldots, x_{n})$ and a variable
order $x_{1}, x_{2}, \ldots, x_{n}$, the number of nodes at the $i$-th level is
given by the number of distinct functions that essentially depend on $x_{i}$ and
result from assigning values to $x_{1}, x_{2}, \ldots, x_{i-1}$ in $f$.
Although this number can be as large as $2^{i-1}$, it can be significantly
smaller for certain functions and can be linked to the linear rank-width of the
underlying graph.
Our approach combines an explicit bound on the number of nodes at each level of
the Feynman decision diagram with a tailored, level-by-level BDD construction algorithm.
By integrating these techniques, we demonstrate that the space and time
complexity of our FeynmanDD framework is exponential in the linear rank-width of
the circuit graph, particularly for two-qubit gate sets such as $\mathcal{T}$
and $\mathcal{G}$.
Weighted model counting extends the method to arbitrary diagonal single-qubit gates because the same BDD can be evaluated for any choice of their phases.
An exact decomposition into diagonal gates and two Hadamards further extends the method to arbitrary single-qubit gates, with a linear rank-width bound of $3w$.

\section*{Results}\label{sec:results}

\subsection*{Quantum circuit simulation with FeynmanDD}\label{subsec:simulation-overview}

FeynmanDD evaluates quantities associated with a quantum circuit by combining a sum-of-powers (SOP) representation with a decision diagram~\cite{WCYJ25}.
We first consider the gate set $\calT=\{\Had,\T,\CZ\}$ and the zero-to-zero amplitude.
For a circuit $C$ with $h$ Hadamard gates, this amplitude has the form
\begin{equation}\label{eq:amplitude-overview}
    \mel{0}{C}{0}
    =\frac{1}{R}\sum_{\vb{x}\in\bin^n}\omega_8^{f(\vb{x})},
\end{equation}
where $R=\sqrt{2^h}$ and $\omega_8=e^{2\pi\ii/8}$.
Here, $\ket{0}$ denotes the all-zero computational basis state on the circuit's qubits, and $n$ is the number of internal SOP variables.
The SOP polynomial $f$ of $C$ takes values in $\ZZ_8$.
Hadamard and $\CZ$ gates contribute quadratic terms, while $\T$ gates contribute linear terms.
The path-sum derivation, the treatment of boundary variables, and the gate encodings are given in Methods.

The evaluation of the SOP can be reduced to counting.
Writing $N_j:=|\{\vb{x}\in\bin^n:f(\vb{x})=j\}|$, we obtain
\begin{equation}\label{eq:residue-overview}
    \frac{1}{R}\sum_{\vb{x}}\omega_8^{f(\vb{x})}
    =\frac{1}{R}\sum_{j=0}^{7}N_j\omega_8^j.
\end{equation}
FeynmanDD represents $f$ by a reduced ordered multi-terminal binary decision diagram (MTBDD)~\cite{BFG+93}.
Each nonterminal node tests a Boolean variable, and each terminal specifies a value of $f$.
The diagram supports the computation of all counts $N_j$ in time $\order{n\size{f}}$, where $\size{f}$ denotes its size.
Thus, the FeynmanDD algorithm consists of obtaining the SOP, constructing its decision diagram, and computing the counts that determine the amplitude.
Acceptance probabilities and sampling tasks can also be reduced to SOP evaluation~\cite{WCYJ25}; the present analysis uses amplitude computation as its running example.

\subsection*{Linear rank-width bounds on FeynmanDD size}\label{subsec:discrete-gate-set}

We analyze the size of FeynmanDD by counting the distinct residual functions obtained by assigning prefixes of a fixed variable ordering.
We relate the number of residual functions to cut-rank and linear rank-width.

For the gate set $\calT$, the SOP polynomial is a degree-2 polynomial over $\ZZ_8$ of the form
$f(\vb{x}) = 4 \sum_{i < j} a_{ij} x_i x_j + \sum_i b_i x_i$, with
$a_{ij} \in \{0, 1\}$ and $b_j \in \ZZ_8$.
This polynomial can be written more compactly as
$f(\vb{x}) = 4 \vb{x}^T \vb{A}_{\lup} \vb{x} + \vb{b}^T \vb{x}$, where
$\vb{b} = (b_i)$ is a vector, and the matrix $\vb{A}_{\lup}$ is defined as
${(\vb{A}_{\lup})}_{ij} = a_{ij}$ for $i < j$ and $0$ otherwise.
We can then define a symmetric matrix $\vb{A} = \vb{A}_{\lup} + \vb{A}_{\lup}^T$
with a zero diagonal, which serves as the adjacency matrix of a graph.
We call this the \emph{variable graph} of the quantum circuit $C$, denoted by
$G_C$.
Its vertex set is $[n]=\{1,\ldots,n\}$, with vertex $i$ corresponding to the SOP variable $x_i$.
An example of a variable graph is shown in \cref{fig:circuit_and_graph}~(d).
Note that the variable graph $G_C$ is associated with a specific SOP
representation and thus with a particular gate set.

\begin{figure}[t]
  \centering
  \includegraphics[width=\textwidth]{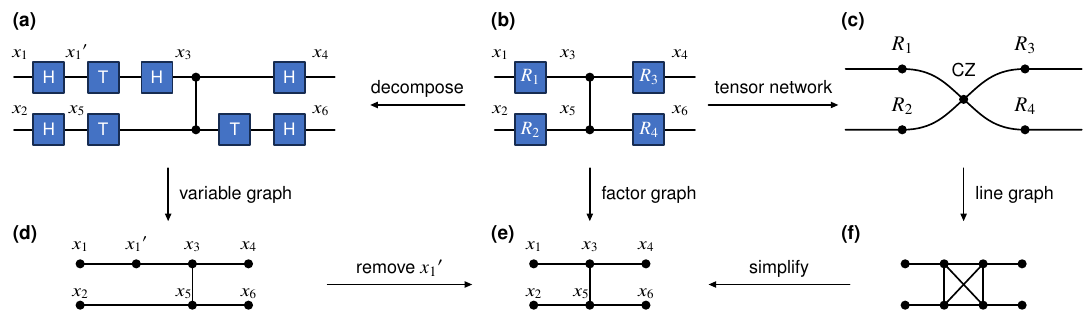}
  \caption{Quantum circuits and graphs.
    (a) A quantum circuit $C'$ constructed from the gate set $\calT$.
    (b) A quantum circuit $C$ constructed from arbitrary single-qubit gates and $\CZ$, which
    has the same unitary as the circuit $C'$.
    Here, we let $R_1 = \Had \T \Had$, $R_2 = \T \Had$, $R_3 = \Had$, and $R_4 = \Had \T$.
    (c) The tensor network of the circuit $C$.
    (d) The variable graph (and the factor graph) of the circuit $C'$.
    (e) The factor graph of the circuit $C$.
    (f) The line graph of the tensor network in (c).
  }\label{fig:circuit_and_graph}
\end{figure}

The cut-rank of a vertex set $X$ in a graph $G=(V,E)$ is defined as
\begin{equation}\label{eq:cut-rank-overview}
    \rho_G(X):=\rank_{\FF_2}\bigl(\vb{A}[X,V\setminus X]\bigr).
\end{equation}
A linear ordering $\pi=(\pi_1,\ldots,\pi_n)$ is a permutation of the vertices of $G$.
For $1\leq i\leq n$, its prefix of length $i$ is the vertex set $\{\pi_1,\ldots,\pi_i\}$.
The width of $\pi$ is the largest cut-rank among these prefixes.
The linear rank-width $\lrw(G)$ is the minimum of this width over all linear orderings.
Full definitions, including rank-width and its relation to other graph parameters, are given in Methods.

The following theorem bounds the size of FeynmanDD in terms of the width of a given variable ordering.

\begin{theorem}\label{thm:bdd-size}
  Suppose we are given a quantum circuit $C$ constructed from the gate set
  $\calT$ with an SOP polynomial
  $f(\vb{x}) = 4 \vb{x}^T \vb{A}_{\lup} \vb{x} + \vb{b}^T \vb{x}$ of $n\geq1$
  variables.
  Let $\pi=(\pi_1,\ldots,\pi_n)$ be any linear ordering of $G_C$ with width $w$, and set $P_i:=\{\pi_1,\ldots,\pi_{i-1}\}$ for $1\leq i\leq n$.
  Under the variable ordering $x_{\pi_1},\ldots,x_{\pi_n}$, the number of nodes at level $i$ is at most $2^{r_i+3}$, where
  $r_i=\rho_{G_C}(P_i)$.
  In particular, the BDD size is at most $n\cdot2^{w+3}$.
\end{theorem}

The proof is given in Supplementary Section~\ref{si:note-size}, which follows the residual-function argument that we sketch here.
For a fixed linear ordering $\pi$ of $G_C$, assigning values to $x_{\pi_1},\ldots,x_{\pi_{i-1}}$ leaves a residual function of $x_{\pi_i},\ldots,x_{\pi_n}$.
Its quadratic terms are fixed at level $i$.
Let $r_i:=\rho_{G_C}(\{\pi_1,\ldots,\pi_{i-1}\})$ be the cut-rank of the assigned prefix.
The assignment changes the linear coefficients by $4$ times a vector in the binary row space of the corresponding cut matrix and contributes a constant phase in $\ZZ_8$.
In a fixed basis of that row space, the residual function is encoded by a residual state $(\gamma_i,c_1,\ldots,c_{r_i})$.
The binary coordinates $(c_1,\ldots,c_{r_i})$ form its cut signature, and $\gamma_i$ records its constant phase.
There are at most $8\cdot2^{r_i}$ such states at that level.

After constructing the MTBDD for $f$ under the variable ordering $x_{\pi_1},\ldots,x_{\pi_n}$, quantum circuit simulation reduces to counting the solutions to $f(\vb{x})=j$ for $0\leq j\leq7$.
This counting takes $\order{n\size{f}}$ time, where $\size{f}$ denotes the BDD size under this ordering.
The size bound is exponential in the width $w$ of $\pi$.
Choosing an ordering of minimum width gives $w=\lrw(G_C)$.

Unary phase terms change the phase component of the states and do not add edges to the variable graph.
Consequently, the bound accommodates arbitrary placements of $\T$ gates on the same $\Had$--$\CZ$ structure.
The analogous modulus-24 encoding for the Google supremacy gate set $\calG$ and its diagram-size bound are given in Supplementary Section~\ref{si:google-size}.

\subsection*{Variable ordering and FeynmanDD construction}\label{subsec:construction-results}

The original FeynmanDD implementation uses heuristics to choose the variable ordering~\cite{WCYJ25}.
This choice can substantially affect the diagram size, but in general, finding an ordering that minimizes the size is $\NP$-hard~\cite{Weg00}.
By \cref{thm:bdd-size}, a linear ordering $\pi$ of $G_C$ with width $w$ gives a FeynmanDD of size at most $n\cdot2^{w+3}$.
Choosing an ordering of minimum width therefore guarantees a size of at most $n\cdot2^{\lrw(G_C)+3}$.

The algorithms of Jeong, Kim, and Oum~\cite{JKO17,JKO21} provide a fixed-parameter tractable (FPT) procedure for obtaining a linear ordering of prescribed width.
Applying this procedure to $G_C$ gives the following lemma, whose proof is given in Supplementary Section~\ref{si:ordering-proof}.

\begin{lemma}\label[lemma]{lem:fpt}
    Given a quantum circuit $C$ constructed from the gate set $\calT$, let $G_C$ be its variable graph on $n\geq1$ vertices and let $w=\lrw(G_C)$ be its linear rank-width.
    Then, in time $2^{\order{w^2}}n^3$, we can output a linear ordering of $G_C$ with width $w$.
\end{lemma}

Note that finding an optimal contraction ordering for tensor networks takes $2^{\order{k^2}}\poly(N)$ time, where $k$ is the treewidth of its $N$-vertex line graph~\cite{MS08,KL23}; this is analogous to the FPT result of \cref{lem:fpt} for linear rank-width.
Once a variable ordering has been chosen, we need to construct the FeynmanDD from the SOP representation.
In Ref.~\citenum{WCYJ25}, the underlying decision diagram is built using the \texttt{apply}
functionality of the decision diagram data structure, and a heuristic method
called the \emph{binary synthesis method} is employed to speed up the
construction.
Although the binary synthesis method supports more gate sets and is practically
efficient, as verified by numerical simulations in Ref.~\citenum{WCYJ25}, its complexity
is difficult to analyze.
Here, we exploit the special quadratic polynomial structure of the SOP and
provide an alternative method for building the decision diagram in a
level-by-level fashion.
This direct method allows us to show that the construction of the decision
diagram is not the bottleneck of the FeynmanDD simulation method.

\Needspace{6\baselineskip}
\begin{theorem}\label{thm:bdd-build-time-complexity}
Given a circuit $C$ over $\calT$, its SOP polynomial $f$ in $n\geq1$ variables, and a variable ordering of width $w$, the corresponding FeynmanDD can be constructed in time $\order{n^2 2^w}$.
\end{theorem}

The construction proceeds by dynamic programming (DP) over the residual states in \cref{thm:bdd-size}.
At a node at level $i$, assigning $x_{\pi_i}=0$ or $1$ gives its two successor residual functions.
Setting $x_{\pi_i}=0$ removes the terms containing that variable.
Setting it to $1$ turns its quadratic terms into linear terms and its linear term into a constant.
Thus, each branch updates the linear coefficients and the constant phase $\gamma_i$.
A table maps each residual state to a unique node, and variables on which a successor does not depend are skipped.
These operations take $\order{n}$ time per node, giving the $\order{n^2 2^w}$ construction bound when combined with \cref{thm:bdd-size}.
The complete algorithm and its correctness and complexity analysis are given in Supplementary Section~\ref{si:construction-proof}.
The counting cost is also $\order{n^2 2^w}$.
The cost of finding an ordering is separate from these bounds for a given ordering.

\subsection*{Extension of FeynmanDD with weighted model counting}\label{subsec:continuous-results}

Here, we extend the original FeynmanDD method with weighted model counting.
FeynmanDD can then simulate circuits over continuous gate sets that include arbitrary single-qubit rotations.
To illustrate the idea, we first reduce the SOP evaluation to weighted model counting.
Consider again circuits over $\calT=\{\Had,\T,\CZ\}$.
Write $f(\vb{x})=4q(\vb{x})+\vb{b}^T\vb{x}$, where $q(\vb{x})=\vb{x}^T\vb{A}_{\lup}\vb{x}$ is evaluated over $\FF_2$, and set $\lambda_i(0)=1$ and $\lambda_i(1)=\omega_8^{b_i}$.
Grouping assignments by their value of $q$ gives the weighted counts
\begin{equation}\label{eq:weighted-counts}
    W_t:=\sum_{\vb{x}:q(\vb{x})=t}\prod_{i=1}^{n}\lambda_i(x_i),
\end{equation}
where $t\in\{0,1\}$.
Each assignment $\vb{x}$ with $q(\vb{x})=t$ contributes the weight $\omega_8^{\vb{b}^T\vb{x}} = \prod_i \lambda_i(x_i)$ to $W_t$.
Thus, $W_t$ is the weighted model count of the Boolean condition $q(\vb{x})=t$~\cite{ChD08,KVD17}.
In a factor graph, $W_t$ is the partition function with one factor enforcing this condition and one weight $\lambda_i(x_i)$ for each variable~\cite{KFL01}.
Since $\omega_8^{f(\vb{x})}=(-1)^{q(\vb{x})}\prod_i\lambda_i(x_i)$, we have $\sum_{\vb{x}}\omega_8^{f(\vb{x})}=W_0-W_1$.
Both counts use the same BDD for $q$, even when the weights are complex.

We now extend this reduction to arbitrary single-qubit diagonal gates.
In this case, the SOP polynomial can still be written as $f(\vb{x})=4q(\vb{x})+\vb{b}^T\vb{x}$, but the linear coefficients $b_i$ can be arbitrary real numbers.
After factoring out an overall phase, $\lambda_i(0)=1$ and $\lambda_i(1)=e^{\pi\ii b_i/4}$ is the product of the relative phases contributed by the diagonal gates on $x_i$.
The remaining sum is again $W_0-W_1$, with both counts evaluated on the same BDD for $q$.
The following theorem uses the cut-rank analysis of \cref{thm:bdd-size} to bound the cost of computing these counts with arbitrary complex unary weights.

\begin{theorem}\label{thm:weighted-counting}
Given the polynomial $q(\vb{x})=\vb{x}^T\vb{A}_{\lup}\vb{x}$ over $\FF_2$ in $n\geq1$ variables and an ordering of its variable graph of width $w$, the weighted counts $W_0,W_1$ can be computed using a BDD for $q$ with at most $2n2^w$ decision nodes.
Constructing the BDD and performing weighted model counting on it take $\order{n^2 2^w}$ arithmetic operations in total.
\end{theorem}

The proof is given in Supplementary Section~\ref{si:weighted-counting}.
Note that the same BDD traversal computes $W_0$ and $W_1$ for cubic and higher-degree Boolean polynomials $q$, and only the linear rank-width argument in \cref{thm:weighted-counting} uses the quadratic form of $q$.

Using the weighted model counting reduction, we now extend FeynmanDD to support the gate set $\calR$ of arbitrary single-qubit unitaries and $\CZ$.
By the standard ZX decomposition~\cite{NC00}, each single-qubit unitary $U$ can be written as $U=D_L\Had D_M\Had D_R$ with diagonal unitaries $D_L,D_M,D_R$.
The diagonal factors enter the weights of $W_t$.
For a circuit over $\calR$, we refer to its simplified factor graph as its variable graph $G_C$ (see Methods).
If the input and output variables of a non-diagonal gate $U$ are both free, the decomposition inserts a new vertex on the edge joining them in $G_C$, as illustrated in \cref{fig:variable_graph}.
The following theorem bounds the width needed to extend an ordering of $G_C$ to the subdivided graph.

\begin{theorem}\label{thm:compilation}
Let $C'$ be obtained from a circuit $C$ over $\calR$ by applying the decomposition $U=D_L\Had D_M\Had D_R$ to each non-diagonal single-qubit gate $U$ whose input and output variables are both free.
Here, $D_L,D_M,D_R$ are diagonal unitaries.
Every width-$w$ ordering of $G_C$ has an efficiently constructible extension to $G_{C'}$ of width at most $3w$.
\end{theorem}

The proof is given in Supplementary Section~\ref{si:note-compilation}.
If $G_C$ has $n\geq1$ vertices, the subdivided graph $G_{C'}$ has at most $2n$ vertices.
Thus, \cref{thm:weighted-counting} gives a total construction and evaluation cost of $\order{n^2 2^{3w}}$ operations.
We remark that the rank-decomposition DP of de Colnet et al.~\cite{dCGH+26} evaluates the weighted SOP directly, while we reduce its evaluation to two weighted model counts and compute both on the BDD for $q$.

\begin{figure}[t]
  \centering
  \includegraphics[width=\textwidth]{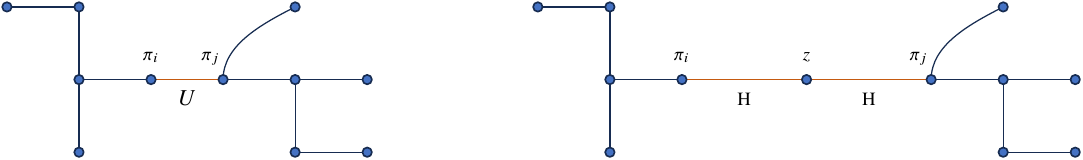}
  \caption{Decomposition of a single-qubit gate and the variable graph.
    The edge $\{\pi_i,\pi_j\}$ joins the input and output variables of $U$ (left).
    Writing $U=D_L\Had D_M\Had D_R$ replaces this edge by the path $\pi_i-z-\pi_j$ (right).
    The diagonal gates contribute unary weights, and $z$ is the only new internal variable.}\label{fig:variable_graph}
\end{figure}

\subsection*{Provable advantage over tensor networks}\label{sec:bounds}

The distinction between linear rank-width and treewidth allows FeynmanDD to simulate circuit families with large tensor-network contraction costs.
The inequality $\lrw(G)\leq\order{\log n\cdot\tw(G)}$ implies that graphs of bounded treewidth admit polynomial-size FeynmanDDs~\cite{AK15,KS93}.
The two widths can also differ substantially: the complete graph on $n$ vertices has linear rank-width $1$ and treewidth $n-1$.
The graph parameters and their relation to tensor-network contraction are defined in Methods.

Here, we demonstrate exponential separations using Instantaneous Quantum Polynomial (IQP) circuits of the form $C=\Had^{\otimes n}D\Had^{\otimes n}$, where $D$ is diagonal~\cite{SB09,BJS11,BMS16}.
The simulation task is to compute the output amplitude and we focus on $\mel{0^n}{C}{0^n}$ without loss of generality.
We compare FeynmanDD with the tensor-network contraction method of Markov and Shi~\cite{MS08}.

The quadratic family uses the gate set $\calT$, with $D$ composed of $\T$ and $\CZ$ gates.
Its construction follows the linear rank-width characterization and produces pairs of circuits with complementary variable graphs.

\begin{theorem}[Quadratic family, informal]\label{thm:quadratic-separation}
There are efficiently constructible pairs of $n$-qubit IQP circuits over $\calT$ for which FeynmanDD computes both zero-to-zero amplitudes in polynomial time.
For at least one circuit in each pair, direct tensor-network contraction requires exponential time and space.
\end{theorem}

The construction uses $k=\order{\log n}$ binary outer products and an optional all-ones term in the corresponding SOP polynomials.
The resulting matrix has rank at most $k+1$, so every cut has cut-rank at most $k+1$ and the FeynmanDD has size $\order{n2^k}$.
The construction and counting bounds then give polynomial-time amplitude evaluation.
The Nordhaus--Gaddum inequality for treewidth~\cite{JW12} ensures that at least one of the two complementary variable graphs has treewidth $\Omega(n)$.
The circuit construction is specified in Methods.
The formal statement and proof of \cref{thm:quadratic-separation} are given in Supplementary Section~\ref{si:note-quadratic}.

Furthermore, we give a linear-network construction whose BDD-size analysis is independent of the linear rank-width characterization.
Knuth's linear-network model~\cite{Knu09} assigns one input variable to each module and bounds BDD size in terms of the numbers of forward and backward Boolean signals communicated between adjacent modules.
This family uses the gate set $\calZ=\{\Had,\Z,\CZ,\CCZ\}$, with $D$ composed of $\Z$, $\CZ$, and $\CCZ$ gates.
Its Boolean phase combines a global quadratic term with cubic terms supported on windows of $\order{\log n}$ consecutive variables.

\begin{theorem}[Linear-network family, informal]\label{thm:linear-network-separation}
There is a family of $n$-qubit IQP circuits over $\calZ$ whose zero-to-zero amplitudes can be computed by FeynmanDD in polynomial time, while direct tensor-network contraction requires exponential time and space.
\end{theorem}

The polynomial-time simulation follows from the linear-network bound on BDD size.
The quadratic part of the circuit supplies a rectangular $\CZ$ block with both dimensions proportional to $n$, which gives the tensor-network contraction lower bound.
The circuit construction is specified in Methods.
The formal statement and proof of \cref{thm:linear-network-separation} are given in Supplementary Section~\ref{si:note-linear-network}.

\begin{figure}[t]
  \centering
  \includegraphics[width=\textwidth]{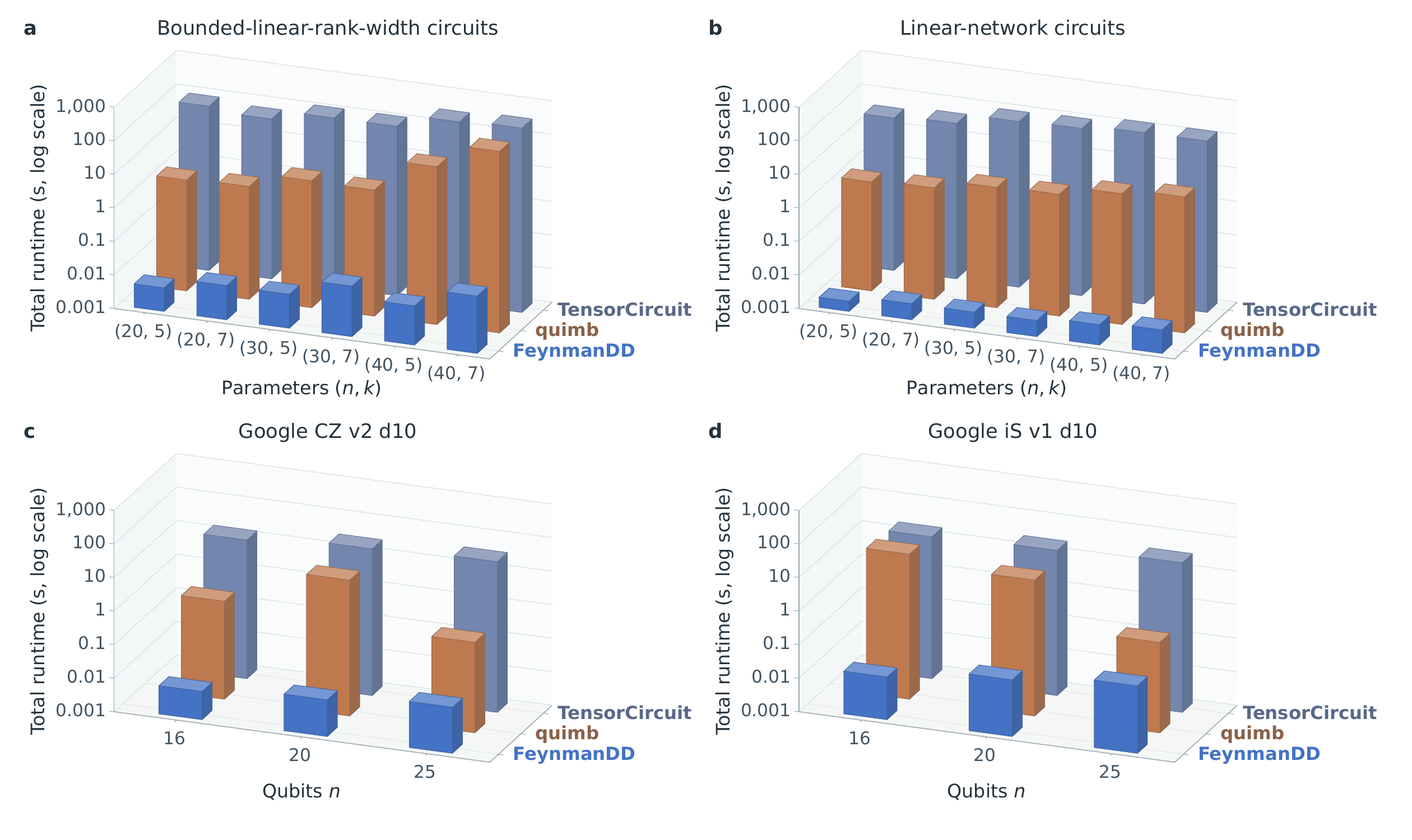}
  \caption{Total runtime of quantum circuit simulation with FeynmanDD, quimb, and TensorCircuit.
  The panels show (a) bounded-linear-rank-width circuits, (b) linear-network circuits, (c) Google CZ v2 d10, and (d) Google iS v1 d10.
  Here, $n$ is the number of qubits, and $k$ is the number of random outer products in (a) or the window size for the cubic terms in (b).
  Bars show mean total runtime in seconds on a logarithmic vertical scale.
  Each parameter configuration contains ten instances.
  For quimb at $(n,k)=(30,5)$ in (a), one instance ran out of memory and the mean uses the nine completed runs.
  The complete numerical results, including contraction times and peak memory usage, are given in Supplementary Section~\ref{si:numerical-benchmarks}.}\label{fig:simulation-runtime}
\end{figure}

\subsection*{Numerical experiments}\label{subsec:exp}

We numerically demonstrate the advantages of FeynmanDD over two state-of-the-art tensor-network simulators, quimb~\cite{Gra18} and TensorCircuit~\cite{ZAW+23}.
The benchmarks cover bounded-linear-rank-width circuits, linear-network circuits, and the Google CZ v2 d10 and iS v1 d10 circuits from \url{https://github.com/sboixo/GRCS}.
When generating circuits with bounded linear rank-width, we require $b_i \notin \{2,4,6\}$ for every entry $b_i$ of the vector $\vb{b}$, since allowing these values can make the 
zero-to-zero amplitude trivial to compute for some instances.
For each parameter configuration, we use ten random circuit instances and record the runtime and peak memory usage of each zero-to-zero amplitude computation.
All three simulators are run on a server with two Intel Xeon Platinum 8358P CPUs (32 cores and 64 threads each) and 512 GiB of memory.

We use quimb's default \texttt{amplitude} method and TensorCircuit's built-in \texttt{amplitude} method with the custom contraction path finder recommended in Section~6.5.1 of Ref.~\citenum{ZAW+23}.
For both tensor-network simulators, we report the contraction time and the total runtime, which includes path finding and other overheads.
For FeynmanDD, we measure runtime from DD construction through counting.
Since heuristics for minimizing linear rank-width are largely unexplored, we use the qubit order of Ref.~\citenum{WCYJ25} and omit the cost of finding a variable ordering from the comparison.
The full benchmark data are provided in the Supplementary Information.
\Cref{fig:simulation-runtime} shows the mean total runtime in seconds, and Supplementary Fig.~\ref{fig:simulation-peak-memory} shows the mean peak memory usage in MB.
We average over the completed runs for each configuration.

FeynmanDD achieves substantial reductions in total runtime and peak memory usage on the bounded-linear-rank-width circuits.
For this family, $k$ is the number of random outer-product terms in \cref{eq:quadratic-family-matrix}.
For $n=40$ and $k=7$, FeynmanDD takes $0.0511$ seconds, compared with $258.46$ seconds for quimb and $312.41$ seconds for TensorCircuit.
Its peak memory usage is $15.79$ MB, compared with $10{,}794.56$ MB for quimb and $102{,}544.06$ MB for TensorCircuit.

FeynmanDD also outperforms both tensor-network simulators on the linear-network circuits.
For these circuits, $k$ is the window size for the local cubic terms, and FeynmanDD uses the binary synthesis implementation to handle the $\CCZ$ gates.
For $n=40$ and $k=7$, the mean total runtime is $0.005$ seconds for FeynmanDD, $11.294$ seconds for quimb, and $129.52$ seconds for TensorCircuit.
The corresponding mean peak memory usage is $12.88$ MB, $1{,}433.34$ MB, and $1{,}293.28$ MB, respectively.
This family was also studied in Ref.~\citenum{WCYJ25}, where it was already shown that FeynmanDD outperformed the DD-based simulators DDSIM~\cite{ZW19,ZHW19} in amplitude computation and sampling and WCFLOBDD~\cite{SCR23,SCR24} in sampling.

Across all six configurations of the Google CZ v2 d10 and iS v1 d10 circuits, FeynmanDD has the shortest mean total runtime and the lowest mean peak memory usage of the three simulators.
For the 25-qubit iS v1 d10 circuits, the mean total runtime is $0.1024$ seconds for FeynmanDD, $0.4945$ seconds for quimb with its default \texttt{auto-hq} optimizer, and $28.70$ seconds for TensorCircuit.
For this family, the decrease in quimb's total runtime with $n$ may reflect a switch from exact path search on smaller simplified tensor networks to a faster heuristic on larger ones.
With the \texttt{greedy} optimizer, quimb's mean total runtime falls to $0.27$--$0.31$ seconds, and FeynmanDD remains faster in all six Google configurations (Supplementary Section~\ref{si:numerical-benchmarks}).
These results demonstrate that FeynmanDD's practical advantages extend beyond the circuit families designed to establish the theoretical separations.

\section*{Discussion}\label{sec:discussion}

In this work, we establish linear rank-width bounds on the complexity of FeynmanDD and demonstrate its advantages over tensor-network simulation through theoretical analysis and numerical experiments.
Weighted model counting extends the method to circuits with arbitrary single-qubit gates and $\CZ$.
For a circuit of $n$ qubits and $m$ gates, Schr\"{o}dinger-type simulation has space and time complexities of $2^{\order{n}}$ and $m 2^{\order{n}}$, respectively.
The Feynman path integral method is space-efficient, requiring only polynomial
space, but its time complexity is $2^{\Omega(m)}$, which is exponential in the
number of gates and typically far worse than the Schr\"{o}dinger method.
Interestingly, however, decision diagrams can somehow automatically organize the
paths to avoid this drawback of Feynman-style simulation and can even outperform
tensor-network based methods for certain circuit families, as our analysis
indicates.

Rank-width and linear rank-width have been used to characterize the entanglement structure and complexity of graph states~\cite{HEB04,DW18,AMDS20}, simulate graph-state computations~\cite{V-d-NDVB07,HI14}, and evaluate Ising partition functions~\cite{HIIL18}.
In the first arXiv version of the present paper (October 2025), we asked whether
the linear rank-width dependence of SOP-based simulation could be improved to
rank-width.
This question has since been answered for the quadratic SOPs of circuits over the $\calT$ gate set.
Specifically, de Colnet et al.~\cite{dCGH+26} give a
dynamic programming algorithm for evaluating amplitudes using $\order{4^w\poly(m)}$ arithmetic operations,
given a rank-decomposition of width $w$ of the SOP variable graph.
Their DP incorporates diagonal gates through unary phase coefficients.
Our weighted model counting reduction instead uses these coefficients as the weights of $W_0$ and $W_1$, evaluated on a BDD for the quadratic part $q$.
Kuyanov and Kissinger~\cite{KK26} obtain a corresponding rank-width bound for
contracting closed graph-like ZX-diagrams.
Codsi and Laakkonen~\cite{CL26} combine graph-width parameters with stabilizer
decompositions, giving a rank-width algorithm with running time $\widetilde{\ordersymbol}(t^{\gamma w})$, where $t$ is the number of non-Clifford gates and $\gamma$ is some constant.

We remark that branching can reduce the width, but joining two child tables gives an exponential factor of $4^{\rw(G_C)}$ in the time bound for SOP evaluation, compared with $2^{\lrw(G_C)}$ for a linear ordering.
In their benchmarks, de Colnet et al.~\cite{dCGH+26} report similar coverage for the branching and linear-layout DPs (407 and 400 of 420 instances, respectively) and leave a direct comparison of their rank-width DP with FeynmanDD to future work.
We conjecture that FeynmanDD can be competitive in practice because a good linear ordering may be easier to find than a rank-decomposition narrow enough to offset the cost of the joins.

Weighted model counting has been applied to circuit
simulation~\cite{MBL24}, expressions in Dirac notation~\cite{vdELLB25},
and path-sum equivalence checking~\cite{HCC+26}.
In artificial intelligence, partition functions play a central role in probabilistic inference on graphical models~\cite{WJ08}, and inference in Bayesian networks can be reduced to weighted model counting~\cite{ChD08}.
In quantum error correction, partition functions determine the relative likelihoods of logical error classes~\cite{DKLP02}.
More recently, quadratic SOPs and weighted model counting have been used to evaluate these partition functions for decoding~\cite{CP26}.

This work leaves open many interesting problems.
First, can we exploit the structure of the SOP polynomials to evaluate them without storing the full decision diagram in memory?
The level-by-level construction in Supplementary Section~\ref{si:construction-proof} suggests a starting point for this optimization.
Second, given the success of tensor-network methods in many-body physics, optimization, and the decoding of quantum error-correcting codes, can FeynmanDD offer advantages in these fields?
Third,  how can we use the cut-rank to guide the search for variable orderings that reduce diagram size and total simulation time in practice?
Finally, can implementation techniques such as GPU acceleration, which have benefited tensor-network methods, also improve decision diagram-based simulation?

\section*{Methods}\label{sec:methods}

\subsection*{Sum-of-powers representations of quantum circuits}\label{subsec:sop-method}

Here, we review the sum-of-powers (SOP) representation of quantum circuits,
which is derived from the Feynman path integral
formalism~\cite{WCYJ25}.
Roughly speaking, an SOP is a tensor that has an alternative representation as a
sum of powers of the $r$-th root of unity $\omega_{r} = e^{2\pi\ii/r}$ for some
fixed $r$.
That is, it takes the form
\begin{equation}\label{eq:sop}
    \frac{1}{R} \sum_{\vb{y}} \omega_{r}^{f(\vb{x}, \vb{y})},
\end{equation}
where $R$ is a normalization coefficient, and $f(\vb{x}, \vb{y})$ is a function
of Boolean variables $\vb{x}, \vb{y}$ taking integer values modulo $r$.
It can be viewed as a tensor with $n = \ell(\vb{x})$ open legs of dimension $2$
where $\ell(\vb{x})$ is the number of Boolean variables in $\vb{x}$.

Many quantities of a quantum circuit can be represented as an SOP if the gate
set has certain properties, which we now describe.
We use the gate set $\calT = \{\Had, \T, \CZ\}$ as an example to
illustrate the idea.
A quantum circuit $C$ specified using $\calT$ can be written as
$C = U_m \ldots U_2 U_1$, where each $U_i \in \calT$ acts on one or two
qubits of the circuit.
Inserting identity operators between the gates, we can write the tensor of $C$
as
\begin{align}
    \mel{\vb{x}_1}{C}{\vb{x}_0}
    = \sum_{\vb{y}_1, \ldots, \vb{y}_{m-1}}
    & \bra{\vb{x}_1} U_m \dyad{\vb{y}_{m-1}} \\
    & U_{m-1} \cdots U_2 \dyad{\vb{y}_1} U_1 \ket{\vb{x}_0}.
\end{align}
For simplicity, we identify $\vb{x}_{0}$ with $\vb{y}_{0}$ and $\vb{x}_{1}$ with
$\vb{y}_{m}$.
The above equation can be written as
\begin{equation}\label{eq:path-sum}
    \mel{\vb{x}_1}{C}{\vb{x}_0} = \sum_{\vb{y}_1, \ldots, \vb{y}_{m-1}}
    \prod_{j=1}^{m} \bra{\vb{y}_j} U_j \ket{\vb{y}_{j-1}}.
\end{equation}
This summation is essentially the Feynman path integral formalism.
If $U_{j}$ is a $\Had$ gate acting on the $s$-th qubit, we have
\begin{equation}
    \bra{\vb{y}_j} U_j \ket{\vb{y}_{j-1}}
    = \frac{1}{\sqrt{2}} {(-1)}^{\vb{y}_{j,s} \vb{y}_{j-1,s}}
    = \frac{1}{\sqrt{2}} \omega_{8}^{4\vb{y}_{j,s} \vb{y}_{j-1,s}},
\end{equation}
which is already in the required power form.
If $U_{j}$ is a $\CZ$ gate acting on the $s$-th and $t$-th qubits, we have
\begin{equation}
    \bra{\vb{y}_j} U_j \ket{\vb{y}_{j-1}}
    = {\omega_{8}}^{4\vb{y}_{j,s} \vb{y}_{j,t}}
    \delta_{\vb{y}_{j,s}, \vb{y}_{j-1,s}} \delta_{\vb{y}_{j,t}, \vb{y}_{j-1,t}},
\end{equation}
which is a power form if we identify the variable $\vb{y}_{j,s}$ with
$\vb{y}_{j-1,s}$, and $\vb{y}_{j,t}$ with $\vb{y}_{j-1,t}$.
The $\T$ gate is also a diagonal gate and can be treated similarly.
When $U_{j}$ is a $\T$ gate on the $s$-th qubit, we have
\begin{equation}
    \bra{\vb{y}_j} U_j \ket{\vb{y}_{j-1}}
    = {\omega_{8}}^{\vb{y}_{j,s}} \delta_{\vb{y}_{j,s}, \vb{y}_{j-1,s}}.
\end{equation}
It is also a power of $\omega_{8}$ if we identify $\vb{y}_{j,s}$ with
$\vb{y}_{j-1,s}$.
This explains how circuits in gate set $\calT$ can be written as an SOP
tensor of modulus $8$.
The polynomial $f$ in this case is a degree-2 polynomial over $\ZZ_8$ and the
normalization factor $R = \sqrt{2^h}$, where $h$ is the number of Hadamard gates
in the circuit.
As an example, consider the circuit in \cref{fig:circuit_and_graph}~(a).
The polynomial $f$ in its SOP representation is
$4 x_1 x_1' + 4 x_1' x_3 + 4 x_2 x_5 + 4 x_3 x_5 + 4 x_3 x_4 + 4 x_5 x_6 + x_1' + 2 x_5$.

The SOP representation is powerful and flexible.
SOP representations exist for other gate sets such as
$\calZ = \{\Had, \Z,\CZ, \CCZ\}$ with modulus $2$ and the Google supremacy gate
set $\calG$~\cite{AAB+19} with modulus $24$.
Apart from computing the amplitude of the quantum circuit, other computational
tasks such as estimating the acceptance probability or sampling the output
distribution can also be reduced to computing the SOP form.
Usually, these quantities can also be expressed in the SOP form
$\frac{1}{R} \sum_{\vb{x}} \omega^{f(\vb{x})}$, but with a different
polynomial $f$ and a different normalization factor $R$.
We refer to a more in-depth discussion on representing circuit-related
quantities using SOP in~\cite{WCYJ25}.
We call the variables that are not summed over in the final representation
external variables, and the rest internal variables.

In this work, we mainly focus on the gate set
$\calT = \{ \Had, \T, \CZ \}$ as it is the most commonly used two-qubit
discrete gate set.
As the above discussion indicates, every $\Had$ gate relates two different
variables of the qubit before and after the application of the gate and can be
seen as a time-like correlation gate.
The $\CZ$ gate uses two input variables identified with two output variables and
can be seen as a space-like correlation gate.
These two gates have the same degree-$2$ form in the polynomial function in the
exponent.
The $\T$ gates are represented as linear terms in the exponent.

\subsection*{Decision diagrams and counting}\label{sec:bdd-definition}

A binary decision diagram (BDD) is a rooted, directed acyclic graph (DAG) used
to represent Boolean functions.
It consists of decision nodes and terminal nodes.
Each decision node is labeled by a variable and has two outgoing edges: one for
the variable being $0$ and the other for the variable being $1$.
The terminal nodes are labeled either $0$ or $1$.
A path from the root to a terminal node corresponds to evaluating the Boolean
function for a given variable assignment.
In the literature, a special type of BDD called a reduced ordered BDD is the
most important, and it is customary to refer to these simply as BDDs.
Here, \emph{ordered} means that the variable ordering along any path from the
root to a terminal node is consistent, and \emph{reduced} means that there are
no isomorphic subgraphs and no node has identical children.
The most important feature of a reduced ordered BDD is that it provides a
canonical representation of a Boolean function.
In our case, we will need a variant of BDDs where the number of terminal nodes
is $r$, not necessarily $2$.
These are called multi-terminal BDDs (MTBDDs)~\cite{BFG+93}; see \cref{fig:feynmandd-overview} for an example.

As explicitly listed in~\cite{Knu09}, BDDs possess many desirable virtues and
properties.
We summarize some of these that will be used in this work; readers are referred
to~\cite{Knu09,Weg00} for more details.
Note that the following properties also hold for MTBDDs with little or no
modification.
For this reason, we may simply refer to MTBDDs as BDDs in the rest of this work.
First, and most important for our work, is their counting ability.
Specifically, given a function $f: \bin^n \to \ZZ_r$ whose BDD representation
has size $\size{f}$, counting the number of solutions to $f(\vb{x}) = j$ for
$j = 0, 1, \ldots, r-1$ can be done in time $\order{n \size{f}}$ using
BDDs~\cite{Weg00,Knu09}.
This serves as the computational engine of the FeynmanDD method.
Second, there is a simple way to bound the number of nodes at each level of a
BDD\@.
Each BDD node can be associated with a subfunction of the original function.
For a node at the $i$-th level, the path from the root to this node activates a
partial assignment of the first $i-1$ variables.
In the notation of~\cite{WCYJ25}, this can be written as
$f[a_1/x_1, \ldots, a_{i-1}/x_{i-1}]$, corresponding to setting $x_j = a_j$ for
$j = 1, \ldots, i-1$.
The resulting subfunction depends only on the remaining variables
$x_i, \ldots, x_n$, and the node represents this subfunction.
Since the BDD is reduced, the number of nodes at the $i$-th level equals the number of distinct residual functions $f[a_1/x_1,\ldots,a_{i-1}/x_{i-1}]$ that essentially depend on $x_i$, as the prefix assignment ranges over $\bin^{i-1}$~\cite{Knu09}.

As shown in~\cite{WCYJ25}, many quantities of quantum circuits can
be reduced to the evaluation of a SOP with no external variables.
The evaluation of the SOP can be written as
\begin{equation}
    \frac{1}{R} \sum_{\vb{x}} \omega_{r}^{f(\vb{x})}
    = \sum_{j=0}^{r-1} \frac{N_{j}}{R} \omega_{r}^{j},
\end{equation}
where $N_{j} = \bigl|\{ x \mid f(x) = j\}\bigr|$.
The FeynmanDD method consists of three steps: (a) identifying the polynomial
$f: \bin^n \to \ZZ_r$ in the SOP\@; (b) constructing the $r$-terminal decision
diagram for $f$; and (c) utilizing the counting ability of BDDs to compute
$N_{j}$ for each $j$ and adding $r$ complex numbers to get an exact numerical
value of the SOP\@.
The resulting multi-terminal decision diagram is called a FeynmanDD to highlight
its origin from the Feynman path integral formalism~\cite{WCYJ25}.
Numerical studies in~\cite{WCYJ25} and the experiments reported here reveal
that it is a competitive method of quantum circuit simulation for various
families of circuits.

To extend the counting step from the residue counts $N_j$ to the weighted counts $W_t$ in \cref{eq:weighted-counts}, we traverse the BDD for the quadratic part $q$ from the terminals to the root~\cite{KVD17}.
For each target value $t \in \bin$, initialize terminal $t$ to one and the other terminal to zero.
At a decision node for $x_i$, multiply each child's value by the corresponding weight $\lambda_i(0)$ or $\lambda_i(1)$ and add the two contributions.
Every variable $x_j$ skipped by an edge contributes a factor $\lambda_j(0)+\lambda_j(1)$.
Apply the same factors to variables preceding the root.
This computes each $W_t$ in $\order{n\size{q}}$ arithmetic operations.

\subsection*{Graph parameters}\label{subsec:graph-method}

Here, we define the graph parameters used in our analysis of FeynmanDD and tensor-network simulation.
Let $G = (V, E)$ be a graph with adjacency matrix
$\vb{A} \in \FF_2^{n\times n}$, and let $X \subset V$ be a subset of its
vertices.
Let $\vb{A}[X, V \backslash X] = (a_{ij}: i \in X, j \in V \backslash X)$ be a
submatrix of $\vb{A}$ of size $|X| \times |V \backslash X|$, formed by the rows
corresponding to the vertices in $X$ and the columns corresponding to the
vertices not in $X$.
Define a rank function $\rho_G$ on vertex subsets of $G$ as
$\rho_G(X) := \rank(\vb{A}[X, V \backslash X])$.
We call $\rho_G$ the cut-rank function of $G$.

A rank decomposition of a graph $G$ is a pair $(T, \mu)$ consisting of a
subcubic tree $T$ and a bijection $\mu: \mathrm{leaf}(T) \to V(G)$~\cite{OS06}.
Here, a subcubic tree is a tree with maximum degree at most $3$.
For every edge $e$ of $T$, removing $e$ from $T$ splits $T$ into two connected
components, which in turn induces a bipartition of the leaves of $T$.
This bipartition of leaves corresponds to a bipartition of the vertices of $G$,
denoted by $X_e$ and $V(G) \setminus X_e$.
The width of an edge $e$ in $T$ is defined as $\rho_G(X_e)$, and the width of a
rank decomposition $(T, \mu)$ is defined as $\max_{e \in E(T)} \rho_G(X_e)$.
The rank-width of a graph $G$, denoted by $\rw(G)$, is the minimum width over
all rank decompositions of $G$, i.e.,
\begin{equation}
    \rw(G) := \min_{(T, \mu)} \max_{e \in E(T)} \rho_G(X_e).
\end{equation}

For linear rank-width, we restrict the subcubic tree $T$ to be a caterpillar,
which is a tree that becomes a path when all its leaves are removed.
However, for our purposes, we use the following equivalent definition.
Let $\pi = (\pi_1, \ldots, \pi_n)$ be a linear ordering of $V(G)$, which is a
permutation of the vertices.
The width of this linear ordering is defined as
$w(\pi) = \max_{1 \leq i \leq n-1} \rho_G(\{ \pi_1, \ldots, \pi_i\})$.
The linear rank-width of $G$, denoted by $\lrw(G)$, is the minimum width over
all linear orderings of $V(G)$, i.e.,
\begin{equation}
    \lrw(G) := \min_{\pi} w(\pi).
\end{equation}

Treewidth measures how far a graph is from a tree, and path-width is its linear variant.
The inequalities
\begin{equation}\label{eq:width-comparison}
    \lrw(G)\leq\pw(G)\leq \order{\log n\cdot\tw(G)}
\end{equation}
follow from~\cite{AK15,KS93}.

\subsection*{Circuit graphs}\label{subsec:circuit-graphs}

We first recall the undirected graphical model in~\cite{BISN18}, where the
underlying graphical model is the factor graph of a quantum circuit.
The factor graph of $C$ also arises from the Feynman path integral formalism.
In \cref{eq:path-sum}, each gate will contribute a factor defined by
$\mel{\vb{y}_j}{U_j}{\vb{y}_{j-1}}$.
For a single-qubit gate $U$, the factor is given by
$\phi_U(x, y) = \mel{y}{U}{x}$, where we use $x, y \in \bin$ to denote the input
and output variables of $U$.
For a two-qubit gate $U$, the factor is given by
$\phi_U(x_1, x_2, y_1, y_2) = \mel{y_1, y_2}{U}{x_1, x_2}$, where $x_1, x_2$ and
$y_1, y_2$ are the input and output variables of $U$, respectively.
Then, the factor graph of $C$ is defined such that each binary variable
corresponds to a vertex in the factor graph and each factor corresponds to a
clique.
The summation over the binary variables in \cref{eq:path-sum} corresponds to
eliminating the vertices and connecting the neighbors of the eliminated vertices
in the factor graph.

Next, we show that the factor graph and tensor network are, in some sense,
equivalent.
The tensor network representation can also be derived from the Feynman path
integral formalism.
Each factor can be viewed as a tensor.
A single-qubit gate is represented by a tensor with two legs, and a two-qubit
gate is represented by a tensor with four legs.
Summation over the binary variables corresponds to connecting the legs of the
tensors, known as tensor contraction.
In \cref{fig:circuit_and_graph}, subfigure~(c) represents the tensor network of
the quantum circuit shown in (b).

Given a tensor network $G$, its \emph{line graph} $G^*$ is constructed with a
precise vertex-edge correspondence.
The vertex set of $G^*$ is in one-to-one correspondence with the edge set of
$G$.
Two vertices $e_1$ and $e_2$ in $G^*$ are connected if there exists a vertex $v$
in $G$ to which both edges are incident.
The line graph for the tensor network in \cref{fig:circuit_and_graph}~(c) is
shown in (f).

From this definition, one can observe that the factor graph defined above is
exactly the line graph of the tensor network representation of $C$.
Moreover, the summation in the Feynman path integral corresponds to contracting
the tensor network or eliminating the vertices in the factor graph, whose
complexity is exponential in the treewidth of the underlying
graph~\cite{MS08,BISN18}.
Since the treewidth of the line graph is linearly related to the treewidth of
the original graph for networks of bounded degree, the undirected graphical model method in~\cite{BISN18} is in
some sense equivalent to the tensor network method.

In~\cite{BISN18} however, a simplified version of the factor graph is actually
used, where diagonal two-qubit gates such as $\CZ$ are represented by factors
with two variables, and diagonal single-qubit gates like $T$ are represented by
factors with one variable~\cite{BISN18}.
In this simplification, the factor for $\CZ$ is identical to that of the
Hadamard gate, up to a multiplicative factor.
In the factor graph, a diagonal single-qubit gate corresponds to a single node,
and a diagonal two-qubit gate corresponds to an edge (see
\cref{fig:circuit_and_graph}~(e)).

The vertices of the simplified factor graph correspond to the internal variables of the path sum.
Repeated $\CZ$ factors on the same two variables cancel in pairs.
For circuits over $\calT$, the resulting factor graph is exactly the variable graph defined by the quadratic part $q$ of the SOP.
Replacing the $\T$ gates by arbitrary diagonal single-qubit gates preserves this graph, since each such gate contributes a factor with at most one variable.
For this reason, we denote the simplified factor graph of a circuit $C$ by $G_C$ and refer to it as the variable graph of $C$.
When the input and output variables of a non-diagonal single-qubit gate $U$ are both internal, the decomposition $U=D_L\Had D_M\Had D_R$ inserts one vertex on the edge joining them.

Despite this graph correspondence, the FeynmanDD method is fundamentally different from the tensor network method or the undirected graphical model method, as the Results demonstrate.

\subsection*{Construction of the IQP circuit families}\label{subsec:family-method}

Here, we describe the two IQP circuit families used to establish the separations in \cref{thm:quadratic-separation,thm:linear-network-separation}.

\subsubsection*{Quadratic family}\label{subsec:lrw-construction}

Let $\vb{v}_0\in\FF_2^n$ be the all-ones vector.
Choose random vectors $\vb{v}_1,\ldots,\vb{v}_k\in\FF_2^n$ and $\vb{b}\in\ZZ_8^n$.
For $a\in\bin$, define the symmetric binary matrix
\begin{equation}\label{eq:quadratic-family-matrix}
    \vb{A}^{(a)}=a\vb{v}_0\vb{v}_0^T+\sum_{j=1}^k\vb{v}_j\vb{v}_j^T.
\end{equation}
The associated polynomial is $f^{(a)}(\vb{x})=4\vb{x}^T\vb{A}^{(a)}_{\lup}\vb{x}+\vb{b}^T\vb{x}$.
It is realized by the circuit $C_a=\Had^{\otimes n}D_a\Had^{\otimes n}$ over $\calT$, where
\begin{equation}\label{eq:quadratic-diagonal-circuit}
    D_a=\prod_{i<j}\CZ_{ij}^{A_{ij}^{(a)}}\prod_i\T_i^{b_i}.
\end{equation}
All gates in $D_a$ commute.
The pair $C_0,C_1$ uses the same vectors $\vb{v}_1,\ldots,\vb{v}_k$ and $\vb{b}$.
Random choices of these vectors generate random circuit pairs.
The amplitude is
\begin{equation}\label{eq:quadratic-family-amplitude}
    \mel{0^n}{C_a}{0^n}=2^{-n}\sum_{\vb{x}\in\bin^n}\omega_8^{f^{(a)}(\vb{x})}.
\end{equation}

Only the off-diagonal entries of $\vb{A}^{(a)}$ specify edges of $G_{C_a}$.
Setting the diagonal to zero leaves every cut-rank and the polynomial $f^{(a)}$ unchanged.
The natural variable ordering $x_1,\ldots,x_n$ can be used directly for the FeynmanDD construction.
The separation in \cref{thm:quadratic-separation} uses $k=\order{\log n}$.

\subsubsection*{Linear-network family}\label{subsec:ln-ckts}

Choose $\vb*{\alpha}\in\bin^n$ and an integer $3\leq k\leq n$.
Let $C_{i:j}$ be a Boolean polynomial consisting of cubic monomials in $x_i,\ldots,x_j$.
Define $A(\vb{x})=\sum_{i=1}^n\alpha_i x_i$ and
\begin{equation}\label{eq:small_bdd_f}
    f(\vb{x})=A(\vb{x})\sum_{i=1}^n x_i+\sum_{i=1}^{n-k+1}C_{i:i+k-1},
\end{equation}
where $f:\bin^n\to\bin$.
All coefficients and sums are over $\FF_2$, and powers of Boolean variables are reduced using $x_i^2=x_i$.
Equal cubic monomials from different windows are combined modulo two.
The IQP circuit corresponding to this polynomial can be constructed with the gate set $\calZ$ and has amplitude
\begin{equation}\label{eq:linear-network-amplitude}
    \mel{0^n}{C}{0^n}=2^{-n}\sum_{\vb{x}\in\bin^n}(-1)^{f(\vb{x})}.
\end{equation}

The diagonal gates follow directly from the multilinear polynomial.
The quadratic part expands as
\begin{equation}\label{eq:linear-network-quadratic-part}
    A(\vb{x})\sum_{i=1}^n x_i
    =\sum_{i=1}^n\alpha_i x_i+
    \sum_{i<j}(\alpha_i+\alpha_j)x_ix_j.
\end{equation}
Let $V_0=\{i:\alpha_i=0\}$ and $V_1=\{i:\alpha_i=1\}$, and write $w=|V_1|$.
For \cref{thm:linear-network-separation}, we use the following gate order.
Apply $\Z_i$ for each $i\in V_1$.
Then, for each $i\in V_0$ in increasing order, apply $\CZ_{ij}$ for all $j\in V_1$ in increasing order.
Finally, apply one $\CCZ$ gate for each cubic monomial with nonzero coefficient.
The $\CZ$ gates form an uninterrupted rectangular block.
Its gate tensor graph contains a $(n-w)\times w$ grid, as proved in Supplementary Section~\ref{si:linear-network-contraction}.
The cubic terms supply the $\CCZ$ gates; the quadratic part uses only Clifford gates.

The polynomial-size decision diagram bound applies to every $\vb*{\alpha}$.
The separation in \cref{thm:linear-network-separation} uses $k=\order{\log n}$ and $\min(w,n-w)=\Omega(n)$.
For example, one may choose any $\vb*{\alpha}$ of weight $\lfloor n/2\rfloor$ and choose the cubic coefficients randomly.
We call this the linear-network family because its BDD-size analysis uses Knuth's model in Fig.~\ref{fig:linear_network}.
The windows of $k$ consecutive variables bound the information passed between its modules.

\section*{Data availability}

The numerical data reported in this study are included in the article and its Supplementary Information.

\section*{Code availability}

The code used in this study is available upon request.

\section*{Author contributions}
Z.J. and J.C. supervised the project.
B.C., Z.W. and Z.J. developed the theory and the algorithm.
Z.W. and R.D. conducted numerical experiments.
All authors contributed to the research and to the writing of the manuscript.

\section*{Competing interests}
The authors declare no competing interests.

\include{supplement}

\end{document}

%% file: supplement.tex
\setcounter{section}{0}
\setcounter{equation}{0}
\setcounter{figure}{0}
\setcounter{table}{0}
\setcounter{algorithm}{0}
\setcounter{stheorem}{0}

\renewcommand{\thesection}{\arabic{section}}
\renewcommand{\theHsection}{S\arabic{section}}
\renewcommand{\thesubsection}{\thesection.\arabic{subsection}}
\renewcommand{\theequation}{S\arabic{equation}}
\renewcommand{\theHequation}{S\arabic{equation}}
\renewcommand{\thealgorithm}{S\arabic{algorithm}}
\renewcommand{\thefigure}{S\arabic{figure}}
\renewcommand{\theHfigure}{S\arabic{figure}}
\renewcommand{\thetable}{S\arabic{table}}
\renewcommand{\theHtable}{S\arabic{table}}
\renewcommand{\figurename}{Fig.}
\renewcommand{\tablename}{Table}
\renewcommand{\refname}{Supplementary References}
\renewcommand{\citenumfont}[1]{S#1}
\renewcommand{\bibnumfmt}[1]{S#1.}
\crefname{table}{Table}{Tables}
\Crefname{table}{Table}{Tables}

\makeatletter
\renewcommand{\maketitle}{\begingroup\raggedright\sffamily
  {\fontsize{20}{24}\selectfont\bfseries Supplementary Information\par}\vspace{6pt}
  {\large\manuscripttitle\par}\vspace{10pt}
  {\fontsize{11}{14}\selectfont\@author\par}\vspace{5pt}
  {\fontsize{8.5}{11}\selectfont\affiliations\par}\endgroup\vspace{12pt}}
\makeatother

\pdfbookmark[0]{Supplementary Information}{supplement}
\maketitle

\section{Bounds on FeynmanDD size}\label{si:note-size}

\subsection{Proof of the FeynmanDD size bound}\label{si:size-proof}

\begin{proof}[Proof of \cref{thm:bdd-size}]
Fix the ordering $\pi$ from the theorem and a level $1\leq i\leq n$.
For a prefix assignment $\vb{a}=(a_1,\ldots,a_{i-1})^T\in\bin^{i-1}$, we use the substitution notation
$f[a_1/x_{\pi_1},\ldots,a_{i-1}/x_{\pi_{i-1}}]$ for the function obtained by setting $x_{\pi_j}=a_j$ for $1\leq j<i$.
We denote this residual function of $x_{\pi_i},\ldots,x_{\pi_n}$ by $f_{\vb{a}}^{(i)}$.
By the subfunction characterization of reduced ordered BDDs~\cite{Knu09}, the number of nodes at level $i$ equals the number of distinct residual functions $f_{\vb{a}}^{(i)}$ that essentially depend on $x_{\pi_i}$.
The same characterization applies to MTBDDs with terminal labels in $\ZZ_8$.
It therefore suffices to bound the number of residual functions over all prefix assignments.

Decompose $f=f_1+f_2$, where $f_2$ consists of the terms involving only the remaining variables and $f_1$ contains all other terms.
Writing $a_{uv}=(\vb{A})_{uv}$ for the entries of the symmetric adjacency matrix, we have
\begin{align}
    f_1(\vb{x})
    &=4\sum_{j=1}^{i-1}\sum_{k=j+1}^{n}a_{\pi_j,\pi_k}x_{\pi_j}x_{\pi_k}
    +\sum_{j=1}^{i-1}b_{\pi_j}x_{\pi_j},\label{eq:function_expand_1}\\
    f_2(\vb{x})
    &=4\sum_{j=i}^{n-1}\sum_{k=j+1}^{n}a_{\pi_j,\pi_k}x_{\pi_j}x_{\pi_k}
    +\sum_{j=i}^{n}b_{\pi_j}x_{\pi_j}.\label{eq:function_expand_2}
\end{align}
Thus, $f_2$ is independent of the prefix assignment.

Let $\vb{M}_i:=\vb{A}[P_i,[n]\setminus P_i]$, with its rows and columns listed in $\pi$ order, and set $\vb{y}:=(x_{\pi_i},\ldots,x_{\pi_n})^T$.
The residual function has the form
\begin{equation}\label{eq:ordered-residual}
    f_{\vb{a}}^{(i)}=f_2+4\bigl(\vb{a}^T\vb{M}_i\bigr)\vb{y}+\gamma_i.
\end{equation}
Here, $\vb{a}^T\vb{M}_i$ is computed over $\FF_2$, and $\gamma_i\in\ZZ_8$ is the value of $f$ after the prefix assignment and with all remaining variables set to zero.
This constant phase includes the contributions of all linear and quadratic terms supported entirely on the assigned prefix.

Fix a basis $\vb{d}_1,\ldots,\vb{d}_{r_i}$ of the binary row space of $\vb{M}_i$.
The vector $\vb{a}^T\vb{M}_i$ has unique binary coordinates $c_1,\ldots,c_{r_i}$ in this basis.
These coordinates form the cut signature of the residual function.
Together with $\gamma_i$, they form the residual state $(\gamma_i,c_1,\ldots,c_{r_i})$.
Since $f_2$ is fixed at level $i$, this state determines the residual function through \cref{eq:ordered-residual}.
There are at most $8\cdot2^{r_i}$ residual states and hence at most $8\cdot2^{r_i}$ distinct residual functions.
The node--subfunction correspondence above then gives at most $2^{r_i+3}$ nodes at level $i$.

By definition, the width of $\pi$ is $w=\max_{1\leq i\leq n}r_i$.
Summing the level bounds gives at most $n\cdot2^{w+3}$ decision nodes.
The same bound also covers the terminal nodes.
For $n\geq2$, the first two levels contain at most one and two decision nodes, respectively, so the total number of nodes is at most $1+2+(n-2)2^{w+3}+8\leq n2^{w+3}$.
For $n=1$, the diagram has at most three nodes.
\end{proof}

\subsection{Extension to the Google gate set}\label{si:google-size}

The gate set $\calG=\{\sqrt{\X},\sqrt{\Y},\sqrt{\W},\iSWAP\}$ has an SOP representation with modulus $24$~\cite{AAB+19,WCYJ25}.
Its gate encodings are listed in \cref{tab:google}.

\begin{table}[t]
\caption{Sum-of-powers representation for the Google supremacy gate
    set.}\label{tab:google}
  \centering
  \begin{tabular}{lllll}
    \toprule
    Gate & Input & Output & Representation & Factor \\
    \midrule
    $\sqrt{\X}$ & $x_{0}$ & $x_{1}$
      & $\omega^{18 x_{0} + 18 x_{1} + 12 x_{0} x_{1}}_{24}$ & $1 / \sqrt{2}$\\
    $\sqrt{\Y}$ & $x_{0}$ & $x_{1}$
      & $\omega^{12 x_{0} + 12 x_{0} x_{1}}_{24}$ & $1 / \sqrt{2}$\\
    $\sqrt{\W}$ & $x_{0}$ & $x_{1}$
      & $\omega^{15 x_{0} + 21 x_{1} + 12 x_{0} x_{1}}_{24}$ & $1 / \sqrt{2}$\\
    $\iSWAP$ & $x_{0}, x_{1}\quad$ & $x_{1}, x_{0}\quad$
      & $\omega^{18 x_{0} + 18 x_{1} + 12 x_{0} x_{1}}_{24}\quad$ & $1$\\
    \bottomrule
  \end{tabular}
\end{table}

The above results also apply to the gate set $\calG$.
The matrix elements of the gates in $\calG$ can be expressed as powers of
$\omega_{24}$ (see \cref{tab:google}).
All gates in $\calG$ are non-diagonal.
From this table, it can be seen that the SOP polynomial for a quantum circuit
constructed from $\calG$ takes the form
$f(\vb{x}) = 12 \vb{x}^T \vb{A}_{\lup} \vb{x} + \vb{b}^T \vb{x}$, where
$\vb{A} \in \FF_2^{n\times n}$ and $\vb{b} \in \ZZ_{24}^n$.
The variable graph of a quantum circuit constructed from $\calG$ can be
similarly defined with $\vb{A}$ as its adjacency matrix.

The row-space representation used for $\calT$ applies because each quadratic coefficient is $r/2$, with $r=24$.
Thus, for $n\geq1$, using the argument for \cref{thm:bdd-size} in the main text, we find that the BDD
size under a fixed variable ordering is at most $24 n 2^{w}$, where $w$ is the
width of the corresponding linear ordering of the variable graph.

\section{Variable ordering and FeynmanDD construction}\label{si:note-construction}

\subsection{Proof of the variable-ordering lemma}\label{si:ordering-proof}

The algorithms of Jeong, Kim, and Oum~\cite{JKO17,JKO21} yield the following constructive result for linear rank-width.

\begin{sproposition}\label[sproposition]{prop:fpt-lrw}
Let $k$ be a positive parameter.
Given an $n$-vertex graph $G$, we can output a linear ordering of $G$ of width at most $k$ or confirm that the linear rank-width of $G$ is larger than $k$ in time $2^{\order{k^2}}n^3$.
\end{sproposition}

\begin{proof}[Proof of \cref{lem:fpt}]
  Using the FPT algorithm in \cref{prop:fpt-lrw}, we can find a linear ordering
  of $G_C$ with width at most $k$ or confirm that the linear rank-width of $G_C$
  is larger than $k$ in time $2^{\order{k^2}}n^3$.
  If $w=0$, the graph is edgeless and any ordering suffices.
  Otherwise, iterate this FPT algorithm for $k = 1, 2$ and so on until a linear ordering is
  output, which happens when $k = w$.
  The total running time is at most $w2^{\order{w^2}}n^3=2^{\order{w^2}}n^3$.
  This procedure returns a linear ordering $\pi=(\pi_1,\ldots,\pi_n)$ of $G_C$ with width $w$.
\end{proof}

\subsection{Construction algorithm and its analysis}\label{si:construction-proof}

We give the full construction for the polynomial $f$ and supplied ordering $\pi=(\pi_1,\ldots,\pi_n)$ of \cref{thm:bdd-build-time-complexity} in the main text.
The construction proceeds top-down, building the decision diagram level by level according to this variable ordering.
At each step, a nonterminal node is visited, and its connections to its child nodes are computed directly from the quadratic structure of $f$.
An SOP with no variables is constant and is represented by a single terminal.

\subsubsection*{Residual functions and stored states}

Use the prefixes $P_i=\{\pi_1,\ldots,\pi_{i-1}\}$ and cut-ranks $r_i$ from \cref{thm:bdd-size} in the main text.
Let $\vb{M}_i=\vb{A}[P_i,[n]\setminus P_i]$, with rows and columns in $\pi$ order, and write $a_{uv}=(\vb{A})_{uv}$.
For a prefix assignment $\vb{a}=(a_1,\ldots,a_{i-1})^T$, the residual function $f_{\vb{a}}^{(i)}=f[a_1/x_{\pi_1},\ldots,a_{i-1}/x_{\pi_{i-1}}]$ has the form established in \cref{eq:ordered-residual}:
\begin{equation}\label{eq:construction-residual}
    f_{\vb{a}}^{(i)}=f_2^{(i)}+4g_i+\gamma_i,
\end{equation}
where $g_i=\sum_{j=i}^{n}\alpha_jx_{\pi_j}$.
Here, $f_2^{(i)}$ is the suffix polynomial $f_2$ in \cref{eq:function_expand_2} and is common to all nodes at level $i$.
The binary coefficient row $\vb*{\alpha}=(\alpha_i,\ldots,\alpha_n)=\vb{a}^T\vb{M}_i$ records the linear correction contributed by the prefix.
The constant $\gamma_i\in\ZZ_8$ equals the residual function evaluated at zero and includes all assigned linear and quadratic terms, together with any global constant in $f$.

For each $\vb{M}_i$, compute the nonzero rows $\vb{d}_{i,1},\ldots,\vb{d}_{i,r_i}$ of its reduced row-echelon form (RREF).
They give the unique expansion
\begin{equation}\label{eq:linear-repre}
    \vb*{\alpha}=\sum_{\ell=1}^{r_i}c_\ell\vb{d}_{i,\ell},
\end{equation}
where $c_\ell\in\FF_2$.
Let $p_{i,\ell}\in\{i,\ldots,n\}$ be the position in $\pi$ of the pivot column of $\vb{d}_{i,\ell}$.
Then, $c_\ell=\alpha_{p_{i,\ell}}$, so the cut signature is read directly from the coefficient row.
The RREF bases are computed incrementally by deleting the first suffix column from the current basis and adjoining the row for the next assigned variable.

A unique table $\calU_i$ at each level indexes nodes by the key $(\gamma_i,c_1,\ldots,c_{r_i})$~\cite{Bry86}.
For deterministic lookup, we use a direct-address array with $2^{r_i+3}$ entries.
Each node stores its level, coefficient row $\vb*{\alpha}$, and constant phase $\gamma_i$ for computing its children.
When $r_i=0$, the key consists only of $\gamma_i$.
At level $n+1$, the residual functions are constants, represented by terminals $t_\gamma$ for $\gamma\in\ZZ_8$.

\subsubsection*{State transitions and node insertion}

Assigning $x_{\pi_i}=\beta$, where $\beta\in\bin$, gives the successor data
\begin{equation}\label{eq:construction-transition}
    \alpha_j^{(\beta)}=\alpha_j+\beta a_{\pi_i,\pi_j}
    \quad\text{and}\quad
    \gamma_{i,\beta}=\gamma_i+\beta(4\alpha_i+b_{\pi_i}).
\end{equation}
The coefficient update applies to $j>i$ and is evaluated over $\FF_2$.
The phase update is modulo eight.
The initial data for the full polynomial are $\vb*{\alpha}=\vb{0}$ and $\gamma_1=f(\vb{0})$.

To obtain a reduced diagram, each nonterminal node is placed at its first essential variable.
For each position $j$, precompute $q_j\in\bin$, which equals one exactly when $a_{\pi_j,\pi_\ell}=1$ for some $\ell>j$.
A residual function starting at position $j$ depends on $x_{\pi_j}$ exactly when $4\alpha_j+b_{\pi_j}\not\equiv0\pmod8$ or $q_j=1$.
The operation $\locate$ finds this position, or level $n+1$ for a constant function.
It uses a suffix view of the coefficient array, advancing its start offset in constant time for each skipped variable.
The operation $\intern$ returns the corresponding node from the unique table or creates it when necessary; at level $n+1$, it returns the terminal $t_\gamma$.

\subsubsection*{Correctness and complexity}

\begin{proof}[Proof of \cref{thm:bdd-build-time-complexity}]
Fix the polynomial $f$ and the supplied ordering $\pi$ of width $w$.
Use the notation and residual representation in \cref{eq:construction-residual}.
The root satisfies it with zero linear correction and phase $f(\vb{0})$.
Substituting $x_{\pi_i}=\beta$ gives \cref{eq:construction-transition}, so each branch represents the required successor function.

At a fixed level, the suffix polynomial $f_2^{(i)}$ is common to all states.
Two residual functions are equal exactly when their phases and binary linear corrections agree.
This follows by comparing their values at the zero assignment and then at assignments with a single nonzero suffix bit.
Since the RREF rows are a basis, equality is equivalent to equality of the keys $(\gamma_i,c_1,\ldots,c_{r_i})$.
Thus, the unique table merges exactly equal residual functions.
The essential-variable test identifies the first variable occurring in a nonzero linear or quadratic term.
Skipping earlier variables preserves the function and gives a valid state at the target level by assigning those variables zero.
The same argument applies to the root.
Induction over the variable positions proves that the returned diagram represents $f$.
Every nonterminal tests an essential variable, so its two children represent different functions.
Together with unique-table merging, this proves that the diagram is reduced and ordered.

We now analyze the running time in the RAM model, with table indices and node addresses stored in single words.
First, we bound the preprocessing cost.
Let $\vb{R}_i$ contain the nonzero rows of the RREF of $\vb{M}_i$.
Deleting its first column and adjoining the row $\vb{A}[\{\pi_i\},[n]\setminus P_{i+1}]$ gives a matrix with the same row space as $\vb{M}_{i+1}$.
The columns follow $\pi$ order.
This matrix has at most $r_i+1$ rows and $n-i$ columns.
Reducing it costs $\order{n(r_i+1)^2}$ binary operations.
Starting with the empty basis for $P_1=\varnothing$, all bases are therefore computed in time $\order{n^2(w+1)^2}$.
The suffix quadratic flags require $\order{n^2}$ additional work.

By \cref{thm:bdd-size} in the main text, there are at most $2^{r_i+3}$ states at level $i$.
Initializing the direct-address tables takes $\order{n2^w}$ time.
For each branch, updating the coefficient row takes $\order{n}$ time.
The essential-variable scan examines at most $n$ positions, using a constant-time test at each position.
The signature at the target level is obtained from the pivot entries of the cached row, as in \cref{eq:linear-repre}, in time $\order{r_j}$ for target level $j$.
Lookup, phase indexing, and storing the new coefficient row are covered by the same $\order{n}$ per-branch bound.
Rank-zero cuts use only the phase, and terminals require constant-time lookup.
Each newly created node is processed once, so the total time is
\begin{equation}\label{eq:construction-total-time}
    \order{n^2(w+1)^2+n^2 2^w}
    =\order{n^2 2^w}.
\end{equation}
\end{proof}

Algorithm~\ref{alg:dd-construction} specifies the complete construction.

\begin{algorithm}[H]
\small
\caption{Construction of a Feynman decision diagram}\label{alg:dd-construction}
\begin{algorithmic}[1]
\Require SOP polynomial $f$ of the form in \cref{thm:bdd-size} in the main text; variable ordering $\pi$.
\Ensure Root of the reduced ordered diagram for $f$ under $x_{\pi_1},\ldots,x_{\pi_n}$.
\State Initialize terminals $t_\gamma$ for $\gamma\in\ZZ_8$.
\If{$f$ is constant with value $\gamma$}
    \State \Return $t_\gamma$.
\EndIf
\State Read $\vb{A}$ and $\vb{b}$ from $f$.
\State Compute the RREF bases and pivots $p_{i,\ell}$ incrementally for $1\leq i\leq n$.
\State Precompute $q_j$ for $1\leq j\leq n$ and initialize all entries of $\calU_1,\ldots,\calU_n$ as empty.
\State $u_{\rm root}\gets\intern(\locate(1,\vb{0},f(\vb{0})))$.
\For{$i=1,\ldots,n$}
    \For{each node $u$ in an occupied entry of $\calU_i$}
        \State Read its coefficient row $\vb*{\alpha}$ and phase $\gamma_i$.
        \For{$\beta\in\bin$}
            \State Compute $(\vb*{\alpha}^{(\beta)},\gamma_{i,\beta})$ using \cref{eq:construction-transition}.
            \State $(j,\vb*{\alpha}',\gamma')\gets\locate(i+1,\vb*{\alpha}^{(\beta)},\gamma_{i,\beta})$.
            \State $v\gets\intern(j,\vb*{\alpha}',\gamma')$ and set the $\beta$-edge of $u$ to $v$.
        \EndFor
    \EndFor
\EndFor
\State \Return $u_{\rm root}$.
\Statex
\Function{$\locate$}{$j,\vb*{\alpha},\gamma$}
    \While{$j\leq n$ and $q_j=0$ and $4\alpha_j+b_{\pi_j}\equiv0\pmod8$}
        \State Advance the start offset of $\vb*{\alpha}$ by one and set $j\gets j+1$.
    \EndWhile
    \State \Return $(j,\vb*{\alpha},\gamma)$.
\EndFunction
\Statex
\Function{$\intern$}{$j,\vb*{\alpha},\gamma$}
    \If{$j=n+1$}
        \State \Return $t_\gamma$.
    \EndIf
    \State $K\gets(\gamma,\alpha_{p_{j,1}},\ldots,\alpha_{p_{j,r_j}})$, with $K=(\gamma)$ if $r_j=0$.
    \If{entry $\calU_j[K]$ is empty}
        \State Create a node $v$ storing $(j,\vb*{\alpha},\gamma)$ and set $\calU_j[K]\gets v$.
    \EndIf
    \State \Return $\calU_j[K]$.
\EndFunction
\end{algorithmic}
\end{algorithm}
\FloatBarrier

For the gate set $\calG$, the same construction uses $12$ in place of $4$ and phases in $\ZZ_{24}$.
The coefficient updates remain binary, and each level table has at most $24\cdot2^{r_i}$ entries.
Thus, the same analysis gives construction time $\order{n^2 2^w}$ for a supplied ordering of width $w$.

\subsection{Weighted model counting}\label{si:weighted-counting}

Here, we prove \cref{thm:weighted-counting} in the main text.

\begin{proof}[Proof of \cref{thm:weighted-counting}]
As in the proof of \cref{thm:bdd-size} in the main text, a residual function at a given level is determined by its cut signature and its constant term.
An ordering of width $w$ gives at most $2^w$ cut signatures per level, and the constant term has only two possible values because $q$ is Boolean.
Thus, the BDD has at most $2\cdot 2^w$ decision nodes per level and at most $2n2^w$ in total.

To construct this BDD, apply Algorithm~\ref{alg:dd-construction} to $f=4q$ over $\ZZ_8$ with $\vb{b}=\vb{0}$.
The variable graph is the same as that of $q$, so the construction takes $\order{n^2 2^w}$ time by \cref{thm:bdd-build-time-complexity} in the main text.
Relabeling the two possible terminal values $0$ and $4$ as $0$ and $1$, respectively, gives the BDD for $q$.

The standard weighted model-counting traversal evaluates a BDD by sums and products, including the factors for skipped variables~\cite{KVD17}.
As described in Methods, it computes each $W_t$ in $\order{n\size{q}}$ arithmetic operations.
The BDD size bound therefore gives $\order{n^2 2^w}$ operations for both counts, completing the proof.
\end{proof}

\section{Weighted model counting for circuits with arbitrary single-qubit gates}\label{si:note-compilation}

Here, we extend the weighted model counting reduction to a circuit $C$ over $\calR$, the gate set of arbitrary single-qubit unitaries and $\CZ$.
We first show how the gates of $C$ determine the weights and the variable graph used to compute $W_0$ and $W_1$.
We then prove \cref{thm:compilation} in the main text by extending a given width-$w$ ordering of $G_C$ to an ordering of the subdivided graph of width at most $3w$.

\subsection{Gate weights and the variable graph}\label{si:single-compilation}

Recall that $W_t=\sum_{\vb{x}:q(\vb{x})=t}\prod_{i=1}^{n}\lambda_i(x_i)$ for $t\in\{0,1\}$, where $q$ is the quadratic part of the SOP over $\FF_2$.
For a diagonal single-qubit gate on an internal variable $x_i$, the two diagonal entries multiply $\lambda_i(0)$ and $\lambda_i(1)$, respectively.
As in Methods, a Hadamard with input $x$ and output $y$ contributes $2^{-1/2}(-1)^{xy}$, and a $\CZ$ gate acting on bits $x$ and $y$ contributes $(-1)^{xy}$.

After identifying the input and output variables of a diagonal gate on each wire, we substitute the bit values fixed by the input and output basis states.
If these states assign different values to the same variable, the amplitude is zero.
For a non-Hadamard single-qubit gate $U$ with fixed input $a$ and free output $y$, the matrix entry $U_{ya}$ gives the weight for $y$.
Similarly, a fixed output $b$ and free input $x$ give the weight $U_{bx}$ for $x$.
If both bits are fixed, $U_{ba}$ contributes a scalar factor.
These substitutions allow zero matrix entries and introduce no new variables.

Now consider a non-diagonal single-qubit gate $U$ whose input variable $x$ and output variable $y$ are both free.
The standard ZX decomposition gives $U=D_L\Had D_M\Had D_R$ with diagonal unitaries $D_L,D_M,D_R$\cite{NC00}.
Writing $D(\beta)$ for the diagonal entry of $D$ at $\beta\in\bin$ and summing over the variable $z$ between the two Hadamards gives
\begin{equation}\label{eq:exact-gate-gadget}
    U_{yx}=\frac12\sum_{z\in\bin}
    D_R(x)D_M(z)D_L(y)(-1)^{xz+zy}.
\end{equation}
The diagonal entries contribute weights for $x$, $z$, and $y$, while the Hadamards contribute the quadratic terms $xz$ and $zy$ to $q$.
Thus, the edge joining $x$ and $y$ in $G_C$ is replaced by the path $x-z-y$, as illustrated in Fig.~\ref{fig:variable_graph} in the main text.
Each such gate introduces one new vertex.

Let $C'$ be the circuit obtained by applying this decomposition as in \cref{thm:compilation} in the main text.
After canceling pairs of $\CZ$ factors on the same variables, the edges of $G_{C'}$ correspond to the nonzero quadratic terms of $q$.
We therefore compute $W_0$ and $W_1$ on the BDD for $q$ with the weights specified above.
For a circuit with $L$ gates, assembling these weights takes $\order{L}$ arithmetic operations.
Multiplying $W_0-W_1$ by the scalar factors and the Hadamard normalization gives the amplitude.
If no free variable remains, the amplitude is the product of the fixed gate matrix entries.

\subsection{Extending the variable ordering}\label{si:simultaneous-compilation}

To control the width of the extended ordering, we first count the edges that are subdivided across a cut of $G_C$.
Let $F$ be the set of edges of $G_C$ that are subdivided to obtain $G_{C'}$.
Each edge in $F$ joins the input and output variables of a single-qubit gate, and the remaining edges of $G_C$ come from $\CZ$ gates.
For a partition $V(G_C)=P\mathbin{\dot\cup}Q$, let $F(P,Q)$ be the set of edges in $F$ with one endpoint in $P$ and the other in $Q$.
All matrix ranks below are over $\FF_2$.

\begin{slemma}\label[slemma]{lem:temporal-cut}
For every partition $V(G_C)=P\mathbin{\dot\cup}Q$,
\begin{equation}\label{eq:temporal-cut}
    |F(P,Q)|\leq 2\rho_{G_C}(P).
\end{equation}
\end{slemma}

\begin{proof}
Orient each edge in $F$ from the input variable of its gate to the output variable.
Each variable is the input of at most one such gate and the output of at most one such gate.
List the edges directed from $P$ to $Q$ as $(x_1,y_1),\ldots,(x_k,y_k)$ in the order of their gates in $C$.
The input variables $x_1,\ldots,x_k$ are distinct, as are the output variables $y_1,\ldots,y_k$.

For $j<\ell$, the gate with input $x_j$ occurs before the gate with output $y_\ell$.
The two variables are never present at the same time, so they cannot be joined by a $\CZ$ gate.
The only non-diagonal single-qubit gate with input $x_j$ has output $y_j\ne y_\ell$.
A gate with input $y_\ell$ and output $x_j$ would have to act after gate $\ell$ and before gate $j$, which is impossible since $j<\ell$.
Thus, the submatrix of the cut matrix with rows $x_1,\ldots,x_k$ and columns $y_1,\ldots,y_k$ is lower triangular with ones on the diagonal.
It follows that $k\leq\rho_{G_C}(P)$.
Interchanging $P$ and $Q$ bounds the number of edges in the reverse direction by $\rho_{G_C}(Q)=\rho_{G_C}(P)$.
Adding the two bounds gives \cref{eq:temporal-cut}.
\end{proof}

\begin{proof}[Proof of \cref{thm:compilation}]
Let $\pi$ be the given ordering of $G_C$ of width $w$.
To construct an ordering $\pi'$ of $G_{C'}$, scan $\pi$ and insert the new vertex of each subdivided edge immediately after its earlier endpoint.
If several new vertices are inserted after the same endpoint, order them arbitrarily.
This construction preserves the order of the original vertices.

Fix a prefix $S'$ of $\pi'$ and let $P=S'\cap V(G_C)$ and $Q=V(G_C)\setminus P$.
Then, $P$ is a prefix of $\pi$.
If the endpoints of a subdivided edge lie on the same side of the cut, its new vertex also lies on that side.
Thus, only subdivisions of edges in $F(P,Q)$ affect the cut matrix.
Write $\vb{M}$ for the cut matrix of $G_C$ at $P$, padded with zero rows and columns for the new vertices on the respective sides of the cut.
Its rank is $\rho_{G_C}(P)$.

Consider an edge in $F(P,Q)$ joining $u\in P$ and $v\in Q$, and let $z$ be its new vertex.
If $z\in S'$, the crossing edge $uv$ is replaced by $zv$, so only column $v$ changes.
If $z\notin S'$, the crossing edge $uv$ is replaced by $uz$, so only row $u$ changes.
In either case, the change to the cut matrix has rank at most one.
For the cut matrix $\vb{M}'$ of $G_{C'}$ at $S'$, subadditivity of rank gives
\begin{equation}\label{eq:compiled-cut-matrix}
    \rank(\vb{M}')\leq \rank(\vb{M})+|F(P,Q)|.
\end{equation}
Applying \cref{lem:temporal-cut} and using the width of $\pi$, we obtain
\begin{equation}\label{eq:compiled-cut-rank}
    \rho_{G_{C'}}(S')
    \leq \rho_{G_C}(P)+|F(P,Q)|
    \leq 3\rho_{G_C}(P)
    \leq 3w.
\end{equation}
This holds for every prefix $S'$, so $\pi'$ has width at most $3w$.
\end{proof}

If $G_C$ has $n\geq 1$ vertices, then $|F|\leq n$ because the gates corresponding to the edges in $F$ have distinct input variables.
The graph $G_{C'}$ therefore has at most $2n$ vertices.
For the extended ordering, \cref{thm:weighted-counting} in the main text gives a total cost of $\order{n^2 2^{3w}}$ arithmetic operations for BDD construction and weighted model counting.

\section{Separations from tensor-network simulation}\label{si:note-separations}

This section gives the formal statements and proofs of \cref{thm:quadratic-separation,thm:linear-network-separation} in the main text.
We focus on computing the output amplitude $\mel{0^n}{C}{0^n}$ for an Instantaneous Quantum Polynomial (IQP) circuit $C=\Had^{\otimes n}D\Had^{\otimes n}$, where $D$ is diagonal.
For the quadratic family, we use the linear rank-width characterization.
For the linear-network family, we use Knuth's model to bound the BDD size.

We compare FeynmanDD with the tensor-network contraction method of Markov and Shi~\cite{MS08}.
The tensor network has one tensor for each gate, one tensor for each fixed input or output, and one binary index for each wire segment.
Let $\Gamma_C$ denote its graph.
We use direct, unsliced contraction of this fixed network, storing intermediate tensors as explicit dense arrays.
In this model, a tensor with $d$ binary indices has $2^d$ entries.
The following definition quantifies the largest intermediate tensor required by a contraction ordering.

\begin{sdefinition}[Contraction complexity]\label[sdefinition]{def:contraction-complexity}
Given a graph $G$ and a contraction ordering $\pi$ of its edges, the contraction complexity $\cc(\pi)$ is the maximum degree of the merged vertices formed in the contraction process.
The contraction complexity of $G$ is the minimum over all contraction orderings: $\cc(G):=\min_{\pi}\cc(\pi)$.
\end{sdefinition}

Markov and Shi~\cite{MS08} relate contraction complexity to treewidth by
\begin{equation}\label{eq:contraction-treewidth}
    \cc(G)=\tw(G^*)\geq\frac{\tw(G)-1}{2},
\end{equation}
where $G^*$ is the line graph of $G$.
Consequently, a lower bound $\cc(\Gamma_C)=\Omega(n)$ implies exponential time and peak storage in this model.
We use this relation to prove the contraction lower bounds for the two circuit families.

\subsection{Quadratic family}\label{si:note-quadratic}

We first consider the gate set $\calT=\{\Had,\T,\CZ\}$ and use the linear rank-width characterization to bound the FeynmanDD size.
Recall that, for $a\in\bin$, the construction uses the symmetric matrix over $\FF_2$
\begin{equation}\label{eq:supp-quadratic-family-matrix}
    \vb{A}^{(a)}=a\vb{v}_0\vb{v}_0^T+\sum_{j=1}^k\vb{v}_j\vb{v}_j^T,
\end{equation}
where $\vb{v}_0$ is the all-ones vector and $\vb{v}_1,\ldots,\vb{v}_k\in\FF_2^n$ are fixed for both choices of $a$.

For a fixed vector $\vb{b}\in\ZZ_8^n$, define the polynomial $f^{(a)}:\bin^n\to\ZZ_8$ by
\begin{equation}\label{eq:quadratic-polynomial}
    f^{(a)}(\vb{x})=4\vb{x}^T\vb{A}^{(a)}_{\lup}\vb{x}+\vb{b}^T\vb{x}.
\end{equation}
Here, $\vb{A}^{(a)}_{\lup}$ is the strictly upper-triangular part of $\vb{A}^{(a)}$.
A corresponding IQP circuit can be constructed as $C_a=\Had^{\otimes n}D_a\Had^{\otimes n}$, with $D_a=\prod_{i<j}\CZ_{ij}^{A_{ij}^{(a)}}\prod_i\T_i^{b_i}$.
The variable graph $G_{C_a}$ has an edge $(i,j)$ exactly when $i\ne j$ and $A_{ij}^{(a)}=1$.
The zero-to-zero amplitude is
\begin{equation}\label{eq:supp-quadratic-family-amplitude}
    \mel{0^n}{C_a}{0^n}=2^{-n}\sum_{\vb{x}\in\bin^n}\omega_8^{f^{(a)}(\vb{x})},
\end{equation}
where $\omega_8=e^{2\pi\ii/8}$.
For this pair of circuits, the following theorem gives the formal statement of \cref{thm:quadratic-separation} in the main text.

\begin{stheorem}[Quadratic-family separation]\label{thm:quadratic-separation-formal}
Let $C_0,C_1$ be the $n$-qubit IQP circuits defined above, constructed from the same vectors $\vb{v}_1,\ldots,\vb{v}_k\in\FF_2^n$ and $\vb{b}\in\ZZ_8^n$.
For each circuit, every variable ordering has width at most $k+1$, the FeynmanDD has size $\order{n2^k}$, and its construction and amplitude evaluation take time $\order{n^2 2^k}$.
At least one circuit in the pair has $\cc(\Gamma_{C_a})=\Omega(n)$.
Thus, for $k=\order{\log n}$, FeynmanDD simulates both circuits in polynomial time, while the tensor-network contraction requires exponential time and space for at least one of them.
\end{stheorem}

\begin{proof}
Each row of $\vb{A}^{(a)}$ is a linear combination of $\vb{v}_0,\vb{v}_1,\ldots,\vb{v}_k$, so $\rank(\vb{A}^{(a)})\leq k+1$.
Every cut-rank is the rank of a submatrix of $\vb{A}^{(a)}$ and is therefore at most $k+1$.
Thus, every variable ordering has width at most $k+1$.
By \cref{thm:bdd-size,thm:bdd-build-time-complexity} in the main text, the FeynmanDD has size $\order{n2^k}$ and can be constructed in time $\order{n^2 2^k}$.
Counting takes $\order{n\size{f^{(a)}}}=\order{n^2 2^k}$ time and determines the amplitude through \cref{eq:supp-quadratic-family-amplitude}.
The matrix and gate list can be generated in $\order{(k+1)n^2}$ time.

Since $\vb{v}_0$ is the all-ones vector, changing $a$ from zero to one complements every off-diagonal entry of $\vb{A}^{(a)}$.
Hence, $G_{C_1}=\overline{G_{C_0}}$.
The Nordhaus--Gaddum inequality for treewidth~\cite{JW12} gives
\begin{equation}\label{eq:quadratic-complement-treewidth}
    \tw(G_{C_0})+\tw(G_{C_1})\geq n-2.
\end{equation}
Thus, at least one of these graphs has treewidth at least $(n-2)/2=\Omega(n)$.
For these IQP circuits, $G_{C_a}$ is a minor of $\Gamma_{C_a}^*$ obtained by contracting the disjoint connected sets of vertices corresponding to the qubit wires.
Therefore, $\cc(\Gamma_{C_a})=\tw(\Gamma_{C_a}^*)\geq\tw(G_{C_a})$.
At least one circuit consequently requires an intermediate tensor with $2^{\Omega(n)}$ entries under every contraction ordering, giving exponential time and space in the stated model.
For $k=\order{\log n}$, FeynmanDD simulates both circuits in polynomial time.
\end{proof}

\FloatBarrier

\subsection{Linear-network family}\label{si:note-linear-network}

We now use the gate set $\calZ=\{\Had,\Z,\CZ,\CCZ\}$ and the linear-network model to bound the BDD size.
Recall that the Boolean phase polynomial $f:\bin^n\to\bin$ is
\begin{equation}\label{eq:linear-network-polynomial}
    f(\vb{x})=A(\vb{x})\sum_{i=1}^n x_i+\sum_{i=1}^{n-k+1}C_{i:i+k-1},
\end{equation}
where $A(\vb{x})=\sum_{i=1}^n\alpha_i x_i$.
Here, $\vb*{\alpha}\in\bin^n$ and $k$ is an integer with $3\leq k\leq n$.
Each $C_{i:j}$ consists of cubic monomials in $x_i,\ldots,x_j$.
All coefficients and sums are over $\FF_2$, and powers of Boolean variables are reduced using $x_i^2=x_i$.
Equal cubic monomials from different windows are combined modulo two.

A corresponding IQP circuit has the form $C=\Had^{\otimes n}D\Had^{\otimes n}$, where the $\Z$, $\CZ$, and $\CCZ$ gates in $D$ correspond to the linear, quadratic, and cubic terms of $f$, respectively.
Its zero-to-zero amplitude is
\begin{equation}\label{eq:supp-linear-network-amplitude}
    \mel{0^n}{C}{0^n}=2^{-n}\sum_{\vb{x}\in\bin^n}(-1)^{f(\vb{x})}.
\end{equation}
As in the quadratic family, counting on the BDD determines this amplitude.

To specify the tensor network, we fix the order of the diagonal gates.
Let $V_0=\{i:\alpha_i=0\}$ and $V_1=\{i:\alpha_i=1\}$, and write $w=|V_1|$.
Apply $\Z_i$ for each $i\in V_1$.
Then, for each $i\in V_0$ in increasing order, apply $\CZ_{ij}$ for all $j\in V_1$ in increasing order.
Finally, apply one $\CCZ$ gate for each cubic monomial with nonzero coefficient.
For this gate order, the following theorem gives the formal statement of \cref{thm:linear-network-separation} in the main text.

\begin{stheorem}[Linear-network separation]\label{thm:linear-network-separation-formal}
Let $C$ be the $n$-qubit IQP circuit with phase polynomial given by \cref{eq:linear-network-polynomial}, where $\vb*{\alpha}\in\bin^n$ and $3\leq k\leq n$.
Its gates have the order specified above.
Write $w=|\vb*{\alpha}|$ and suppose that $2\leq w\leq n-2$.
Its Boolean phase has a BDD of size $\order{n4^k}$, and the BDD can be constructed and the amplitude evaluated in time $\order{\poly(n,k)4^k}$.
For the specified gate tensor network,
\begin{equation}\label{eq:linear-network-contraction-bound}
    \cc(\Gamma_C)\geq\frac{\min(w,n-w)-1}{2}.
\end{equation}
Consequently, when $k=\order{\log n}$ and $\min(w,n-w)=\Omega(n)$, FeynmanDD simulates $C$ in polynomial time, while the tensor-network contraction requires exponential time and space.
\end{stheorem}

\subsubsection{FeynmanDD size and simulation complexity}\label{si:linear-network-size}

We use Knuth's linear-network model~\cite{Knu09} to bound the BDD size.
Consider a Boolean function $f$ computed by the network in Fig.~\ref{fig:linear_network}.
Each module $M_j$ receives $x_j$ and communicates with its two neighbours through Boolean signals.
There are $a_j$ forward signals from $M_j$ to $M_{j+1}$ and $b_j$ backward signals from $M_{j+1}$ to $M_j$.
Set $a_0=b_0=b_n=0$ and $a_n=1$.
The signals are required to be uniquely determined by every input assignment.
The final forward signal is $f$.

\begin{stheorem}[Theorem M of~\cite{Knu09}]\label{thm:linear_network_BDD}
If a Boolean function $f$ is computed by such a linear network, then
\begin{equation}\label{eq:knuth-linear-network-bound}
    \size{f}\leq\sum_{j=0}^n2^{a_j2^{b_j}}.
\end{equation}
\end{stheorem}

\begin{figure}[t]
  \centering

  \resizebox{\linewidth}{!}{\begin{tikzpicture}[node distance=1.5cm, >=latex]

    \node[gate, minimum width=.8cm, minimum height=2cm] (M1) at (0,0) {$M_1$};
    \node[gate, minimum width=.8cm, minimum height=2cm, right of=M1] (M2) {$M_2$};
    \node[gate, minimum width=.8cm, minimum height=2cm, right of=M2] (M3) {$M_3$};
    \node[gate, minimum width=.8cm, minimum height=2cm, right of=M3] (M4) {$M_4$};
    \node[gate, minimum width=.8cm, minimum height=2cm, right = 4cm of M4] (Mn) {$M_n$};

    \node[above = 0.5cm of M1] (x1) {$x_1$};
    \node[above = 0.5cm of M2] (x2) {$x_2$};
    \node[above = 0.5cm of M3] (x3) {$x_3$};
    \node[above = 0.5cm of M4] (x4) {$x_4$};
    \node[above = 0.5cm of Mn] (xn) {$x_n$};

    \node[right = .5cm of Mn] (output) {Output};

    \draw[->] (x1) -- (M1.north);
    \draw[->] (x2) -- (M2.north);
    \draw[->] (x3) -- (M3.north);
    \draw[->] (x4) -- (M4.north);
    \draw[->] (xn) -- (Mn.north);

    \draw[->] ([yshift=.8cm]M1.east) -- ([yshift=.8cm]M2.west);
    \draw[->] ([yshift=.4cm]M1.east) -- ([yshift=.4cm]M2.west);
    \draw[->] ([yshift=.8cm]M2.east) -- ([yshift=.8cm]M3.west);
    \draw[->] ([yshift=.55cm]M2.east) -- ([yshift=.55cm]M3.west);
    \draw[->] ([yshift=.3cm]M2.east) -- ([yshift=.3cm]M3.west);
    \draw[->] ([yshift=.8cm]M3.east) -- ([yshift=.8cm]M4.west);
    \draw[->] ([yshift=.6cm]M3.east) -- ([yshift=.6cm]M4.west);
    \draw[->] ([yshift=.4cm]M3.east) -- ([yshift=.4cm]M4.west);
    \draw[->] ([yshift=.2cm]M3.east) -- ([yshift=.2cm]M4.west);
    \draw[->] ([yshift=.8cm]M4.east) -- ([yshift=.8cm,xshift=.7cm]M4.east);
    \draw[->] ([yshift=.625cm]M4.east) -- ([yshift=.625cm,xshift=.7cm]M4.east);
    \draw[->] ([yshift=.45cm]M4.east) -- ([yshift=.45cm,xshift=.7cm]M4.east);
    \draw[->] ([yshift=.275cm]M4.east) -- ([yshift=.275cm,xshift=.7cm]M4.east);
    \draw[->] ([yshift=.1cm]M4.east) -- ([yshift=.1cm,xshift=.7cm]M4.east);
    \draw[->] ([yshift=.8cm,xshift=-.7cm]Mn.west) -- ([yshift=.8cm]Mn.west);
    \draw[->] ([yshift=.55cm,xshift=-.7cm]Mn.west) -- ([yshift=.55cm]Mn.west);
    \draw[->] ([yshift=.3cm,xshift=-.7cm]Mn.west) -- ([yshift=.3cm]Mn.west);

    \draw[->] ([yshift=-.8cm]M2.west) -- ([yshift=-.8cm]M1.east);
    \draw[->] ([yshift=-.5cm]M2.west) -- ([yshift=-.5cm]M1.east);
    \draw[->] ([yshift=-.8cm]M3.west) -- ([yshift=-.8cm]M2.east);
    \draw[->] ([yshift=-.5cm]M3.west) -- ([yshift=-.5cm]M2.east);
    \draw[->] ([yshift=-.8cm]M4.west) -- ([yshift=-.8cm]M3.east);
    \draw[->] ([yshift=-.5cm]M4.west) -- ([yshift=-.5cm]M3.east);
    \draw[->] ([yshift=-.8cm,xshift=.7cm]M4.east) -- ([yshift=-.8cm]M4.east);
    \draw[->] ([yshift=-.5cm,xshift=.7cm]M4.east) -- ([yshift=-.5cm]M4.east);
    \draw[->] ([yshift=-.8cm]Mn.west) -- ([yshift=-.8cm,xshift=-.7cm]Mn.west);
    \draw[->] ([yshift=-.5cm]Mn.west) -- ([yshift=-.5cm,xshift=-.7cm]Mn.west);

    \draw [decorate,thick,decoration={brace,amplitude=3pt,mirror,raise=.8cm}]
    ([yshift=.1cm]M4.east) -- ([yshift=.8cm]M4.east) node[midway,xshift=1.2cm]{$a_{4}$};

    \draw [decorate,thick,decoration={brace,amplitude=3pt,mirror,raise=.8cm}]
    ([yshift=-.9cm]M4.east) -- ([yshift=-.4cm]M4.east) node[midway,xshift=1.2cm]{$b_{4}$};

    \draw [decorate,thick,decoration={brace,amplitude=3pt,raise=.8cm}]
    ([yshift=.3cm]Mn.west) -- ([yshift=.8cm]Mn.west) node[midway,xshift=-1.4cm]{$a_{n-1}$};

    \draw [decorate,thick,decoration={brace,amplitude=3pt,raise=.8cm}]
    ([yshift=-.9cm]Mn.west) -- ([yshift=-.4cm]Mn.west) node[midway,xshift=-1.4cm]{$b_{n-1}$};

    \node[draw=none] (dots1) at ($(M4)!0.45!(Mn)$) {\dots};
    \node[draw=none] (dots1) at ([yshift=1.6cm]$(M4)!0.45!(Mn)$) {\dots};
    \node[draw=none] (dots1) at ([yshift=.8cm]$(M4)!0.45!(Mn)$) {\dots};
    \node[draw=none] (dots1) at ([yshift=-.8cm]$(M4)!0.45!(Mn)$) {\dots};

    \draw[->] (Mn.east) -- (output.west);

  \end{tikzpicture}}
  \caption{A linear network for computing Boolean functions.
    Figure adapted from~\cite{Knu09}.}\label{fig:linear_network}
\end{figure}

The following proposition applies \cref{thm:linear_network_BDD} to the polynomial in \cref{eq:linear-network-polynomial}.

\begin{sproposition}\label[sproposition]{prop:linear-network-size}
For every $\vb*{\alpha}$ and every choice of the cubic terms in \cref{eq:linear-network-polynomial}, the natural variable ordering gives a BDD of size $\order{n4^k}$.
The diagram can be constructed and the amplitude evaluated in time $\order{\poly(n,k)4^k}$.
\end{sproposition}

\begin{proof}
  Let $A_j(x) = \sum_{i=j}^{n} \alpha_i x_i$ represent a partial sum of $A(x)$
  from index $j$ to $n$, with $A_1(x) = A(x)$.
  We demonstrate that $f$ can be computed using the linear network in
  \cref{tab:construction}.
  The computation proceeds through a systematic signal transmission process as follows.
  For backward signals, module $M_j$ sends $A_j$ to module $M_{j-1}$ when $j > 1$.
  This enables module $M_1$ to compute $A(x)$ and $x_1 A(x)$.
  Subsequently, each module $M_j$ receives $A(x)$ and
  $A(x) \sum_{i=1}^{j-1} x_i$ from $M_{j-1}$.
  It then computes and sends $A(x) \sum_{i=1}^{j} x_i$ to $M_{j+1}$.
  Through this signal propagation strategy, the first term of $f(x)$ in
  \cref{eq:linear-network-polynomial} can be computed across the linear network.

  \begin{table}[t]
  \caption{Linear-network construction for a function $f$ with small BDD
      size.}\label{tab:construction}
    \centering
    \begin{tabular}{lll}
      \toprule
      Module & Forward Signals & Backward Signal\\
      \midrule
      $M_1$ & $x_1, x_1 A(x), A(x) \vphantom{\sum\limits_{i=1}^{k}}$ & \\
      $M_2$ & $x_1, x_2, A(x) \sum\limits_{i=1,2} x_i, A(x)$ & $A_2$ \\
      $\;\;\vdots$ & $\quad\vdots$ & $\;\;\vdots$ \\
      $M_{k-1}\qquad$ & $x_1, \ldots, x_{k-1}, A(x) \sum\limits_{i=1}^{k-1} x_i, A(x)$ & $ A_{k-1}$ \\
      $M_k$ & $x_2, \ldots, x_k, C_{1:k}, A(x) \sum\limits_{i=1}^k x_i, A(x)$ & $A_k$ \\
      $\;\;\vdots$ & $\quad\vdots$ & $\;\;\vdots$ \\
      $M_{n-1}$ & $x_{n-k+1}, \ldots, x_{n-1}, \sum\limits_{i=1}^{n-k} C_{i : i+k-1},
                  A(x) \sum\limits_{i=1}^{n-1} x_i, A(x)\qquad$ & $A_{n-1}$ \\
      $M_n$ & $f(x)$ & $A_n$\\
      \bottomrule
    \end{tabular}
  \end{table}
  For the second term, the computation follows a progressive signal transmission
  strategy. 
  When $j < k$, module $M_j$ sends variables $x_1, \ldots, x_j$ to
  module $M_{j+1}$.
  At $j = k$, module $M_k$ receives variables $x_1, \ldots, x_k$ and computes
  $C_{1:k}$.
  Since $x_1$ becomes unnecessary for subsequent computations, $M_k$ sends
  $x_2, \ldots, x_k$ and $C_{1:k}$ to $M_{k+1}$.
  When $j = k+1$, module $M_{k+1}$ computes $C_{2:k+1}$.
  As $x_2$ is no longer required, $M_{k+1}$ forwards $x_3, \ldots, x_{k+1}$ and
  the sum $C_{1:k} + C_{2:k+1}$ to $M_{k+2}$.
  This process continues, progressively eliminating unnecessary
  variables and accumulating computational results, ultimately enabling the
  computation of $f(x)$.

  Since each module $M_j$ has at most $k+2$ forward signals and $1$ backward
  signal, the BDD size of $f$ is bounded by $\size{f} = \order{n 4^k}$ according
  to \cref{thm:linear_network_BDD}.

  Each residual function is a degree-three multilinear polynomial over $\FF_2$, so its canonical coefficient representation can be updated and compared in polynomial time.
  Standard BDD construction and counting therefore take time $\order{\poly(n,k)4^k}$.
\end{proof}

\FloatBarrier

\subsubsection{Tensor-network contraction lower bound}\label{si:linear-network-contraction}

We now bound the tensor-network contraction cost using the $\Z$ and $\CZ$ gates.
Consider the subcircuit formed by the following terms of $f$ in \cref{eq:linear-network-polynomial},
\begin{equation}\label{eq:supp-linear-network-quadratic-part}
    \sum_{i:\alpha_i=1}x_i+
    \sum_{\substack{i:\alpha_i=1\\j:\alpha_j=0}}x_ix_j,
\end{equation}
which consists of $\Z$ and $\CZ$ gates.
The prescribed gate order places these $\CZ$ gates consecutively, after all $\Z$ gates and before all $\CCZ$ gates.
\begin{figure}[htbp]
  \centering
  \resizebox{\linewidth}{!}{\begin{tikzpicture}[node distance=.4cm,
    circ/.style={circle, draw, fill=ChannelColor, minimum width=.2cm, inner sep=0}]
    \node[] (In1) at (0,0) {};
    \node[above=of In1] (In2) {};
    \node[above=of In2] (In3) {};
    \node[above=of In3] (In4) {};
    \node[above=of In4] (In5) {};

    \node[] (Out1) at (5.5,0) {};
    \node[above=of Out1] (Out2) {};
    \node[above=of Out2] (Out3) {};
    \node[above=of Out3] (Out4) {};
    \node[above=of Out4] (Out5) {};

    \draw (In1) -- (Out1);
    \draw (In2) -- (Out2);
    \draw (In3) -- (Out3);
    \draw (In4) -- (Out4);
    \draw (In5) -- (Out5);

    \node[control, right=of In1, label={120:$i$}] (i) {};
    \node[control, right=of In2, label={120:$h$}] (h) {};
    \node[control, right=of In4, label={120:$g$}] (g) {};
    \node[control, right=1cm of In4] (ap) {};
    \node[control, right=1cm of In5, label={120:$a$}] (a) {};
    \node[control, right=1.6cm of In5, label={120:$b$}] (b) {};
    \node[control, right=1.6cm of In2] (bp) {};
    \node[control, right=2.2cm of In5, label={120:$c$}] (c) {};
    \node[control, right=2.2cm of In1] (cp) {};

    \node[control, right=3.2cm of In4, label={120:$d$}] (d) {};
    \node[control, right=3.2cm of In3] (dp) {};
    \node[control, right=3.8cm of In3, label={120:$e$}] (e) {};
    \node[control, right=3.8cm of In2] (ep) {};
    \node[control, right=4.4cm of In3, label={120:$f$}] (f) {};
    \node[control, right=4.4cm of In1] (fp) {};

    \draw (a.center) -- (ap.center);
    \draw (b.center) -- (bp.center);
    \draw (c.center) -- (cp.center);
    \draw (d.center) -- (dp.center);
    \draw (e.center) -- (ep.center);
    \draw (f.center) -- (fp.center);

    \draw[->, thick] ([xshift=.5cm]Out3.center) -- ([xshift=1.5cm]Out3.center);

    \begin{scope}[xshift=7cm, node distance=.6cm]

      \node[circ, label={60:$a$}] (A) at (1.5,1.1) {};
      \node[circ, right=of A, label={60:$b$}] (B) {};
      \node[circ, right=of B, label={60:$c$}] (C) {};

      \node[circ, label={60:$d$}] (D) at (1.5,0) {};
      \node[circ, right=of D, label={60:$e$}] (E) {};
      \node[circ, right=of E, label={60:$f$}] (F) {};

      \node[circ, label={60:$g$}] (G) at (1.0,1.8) {};

      \node[circ, label={60:$h$}] (H) at (2.1,2.0) {};
      \node[circ, label={60:$i$}, right=1cm of H] (I) {};

      \draw (A) -- (B) -- (C) -- (F) -- (E) -- (D) -- (A);
      \draw (B) -- (E);

      \draw (B) edge [out=135,in=-45] (H);
      \draw (C) edge [out=135,in=-45] (I);

      \draw (H) edge [out=135, in=-45] (1.7,2.3);
      \draw (I) edge [out=135, in=-45] (2.8,2.4);

      \draw (A) edge [out=135, in=-15] (G);

      \draw [bend left=10] (A) to ([shift=({-.6,-.7})]A);
      \draw [bend left=10] (D) to ([shift=({-.6,-.7})]D);
      \draw [bend right=20] (D) to ([shift=({.3,-.7})]D);
      \draw [bend right=20] (E) to ([shift=({.3,-.7})]E);
      \draw [bend right=10] (F) to ([shift=({.7,-.6})]F);
      \draw [bend left=15] (F) to ([shift=({.6,.1})]F);
      \draw [bend left=15] (C) to ([shift=({.8,.1})]C);
      \draw [bend right=10] (G) to ([shift=({-.5,.2})]G);
    \end{scope}

  \end{tikzpicture}}
  \caption{Left: the $\Z$ gates and rectangular $\CZ$ block for $\vb*{\alpha}=(0,1,0,1,1)$ using the gate order specified above.
The qubits are ordered from top to bottom.
Right: the corresponding part of the gate tensor network.
The six $\CZ$ vertices $a,b,c,d,e,f$ form a $2\times3$ grid.
The cubic gates follow this block.}\label{fig:lattice}
\end{figure}

We use the following monotonicity property.

\begin{sproposition}\label[sproposition]{prop:subgraph_cc}
If $G'$ is a subgraph of $G$, then $\cc(G')\leq\cc(G)$.
\end{sproposition}

\begin{proof}
The line graph $(G')^*$ is a subgraph of $G^*$.
Treewidth is monotone under taking subgraphs.
The identity $\cc(G)=\tw(G^*)$~\cite{MS08} gives the result.
\end{proof}

\begin{sproposition}\label[sproposition]{prop:linear-network-contraction}
For the linear-network circuit with the gate order specified above, suppose that $2\leq w\leq n-2$.
Then, $\Gamma_C$ contains an $(n-w)\times w$ grid as a subgraph and
\begin{equation}\label{eq:linear-network-grid-lower-bound}
    \cc(\Gamma_C)\geq\frac{\min(w,n-w)-1}{2}.
\end{equation}
\end{sproposition}

\begin{proof}
Write $V_0=\{i_1<\cdots<i_{n-w}\}$ and $V_1=\{j_1<\cdots<j_w\}$.
Let $g_{r,s}$ be the tensor vertex of the gate $\CZ_{i_rj_s}$.
For fixed $r$, the gates $g_{r,1},\ldots,g_{r,w}$ occur consecutively on wire $i_r$.
Hence, $\Gamma_C$ contains the edges $g_{r,s}g_{r,s+1}$.
For fixed $s$, the gates $g_{1,s},\ldots,g_{n-w,s}$ occur consecutively on wire $j_s$.
Hence, it also contains the edges $g_{r,s}g_{r+1,s}$.
No $\Z$ or $\CCZ$ gate lies between these gates on their wires.
The vertices $g_{r,s}$ and these edges therefore form an $(n-w)\times w$ grid $L$.
Fig.~\ref{fig:lattice} illustrates the case $n=5$ and $w=3$.

The grid has treewidth $\tw(L)=\min(w,n-w)$.
By \cref{prop:subgraph_cc} and \cref{eq:contraction-treewidth},
\begin{equation}
    \cc(\Gamma_C)\geq\cc(L)
    \geq\frac{\tw(L)-1}{2}
    =\frac{\min(w,n-w)-1}{2}.
\end{equation}
\end{proof}

\begin{proof}[Proof of \cref{thm:linear-network-separation-formal}]
The FeynmanDD bounds follow from \cref{prop:linear-network-size}, and the contraction bound for tensor network follows from \cref{prop:linear-network-contraction}.
When $k=\order{\log n}$, construction and amplitude evaluation for FeynmanDD take polynomial time.
When $\min(w,n-w)=\Omega(n)$, every contraction ordering in the specified model requires a dense tensor with $2^{\Omega(n)}$ entries.
This gives exponential time and peak storage.
Choosing $w=\lfloor n/2\rfloor$ provides the circuit family in \cref{thm:linear-network-separation} in the main text.
\end{proof}

\section{Additional numerical data}\label{si:numerical-benchmarks}

Here, we provide additional data for the comparison of FeynmanDD, quimb, and TensorCircuit in the main text.
Table~\ref{tab:lrw_results} lists the mean total runtime and peak memory usage for each benchmark configuration, along with the contraction times for quimb and TensorCircuit.
Table~\ref{tab:google-greedy} presents the results obtained using quimb with the \texttt{greedy} optimizer.
Figure~\ref{fig:simulation-peak-memory} compares the mean peak memory usage of the three simulators across the four circuit families.

\begin{table}[htbp]
\caption{Quantum circuit simulation benchmarks on Google circuits, bounded-linear-rank-width circuits, and linear-network circuits.
  Here, $n$ is the number of qubits, and $k$ specifies the circuit-family parameter described in the Numerical experiments subsection of the main text.
  A dash indicates that $k$ does not apply.
  Each parameter configuration uses 10 circuit instances.
  Time is measured in seconds (s), and memory usage is measured in MB\@.
  The timing data of tensor-network simulators is presented as contraction time/total time.
  The {\color{red} ($x$)} mark indicates that {\color{red} $x$} tests in the
  group of ten circuits throws out-of-memory (OOM) and the displayed value is
  derived from the average of the remaining $10 - {\color{red} x}$ test
  results.}\label{tab:lrw_results}
\centering\footnotesize\setlength{\tabcolsep}{3pt}
\begin{tabular*}{\linewidth}{@{\extracolsep{\fill}}c c c c r c c r r@{}}
\toprule
 & & & \multicolumn{2}{c}{FeynmanDD} & \multicolumn{2}{c}{quimb} & \multicolumn{2}{c}{TensorCircuit} \\
\cmidrule(lr){4-5} \cmidrule(lr){6-7} \cmidrule(lr){8-9}
n & k & Gates & Time & Mem & Time & Mem & Time & Mem$\quad$ \\
\midrule
\multicolumn{9}{c}{Google CZ v2 d10} \\
16 & \textemdash & 115 & 0.0069 & 5.73 & 0.0015/0.82200 & 196.75 & 0.46/13.53 & 418.82 \\
20 & \textemdash & 145 & 0.0123 & 7.13 & 0.0012/11.1465 & 204.42 & 0.59/23.39 & 427.34 \\
25 & \textemdash & 184 & 0.0235 & 9.66 & 0.0013/0.49820 & 200.52 & 0.68/30.87 & 434.68 \\
\midrule
\multicolumn{9}{c}{Google iS v1 d10} \\
16 & \textemdash & 115 & 0.0186 & 8.41 & 0.0012/20.9540 & 204.62 & 0.46/16.92 & 425.38 \\
20 & \textemdash & 145 & 0.0485 & 15.80 & 0.0012/11.1722 & 201.10 & 0.57/21.13 & 432.32 \\
25 & \textemdash & 184 & 0.1024 & 28.03 & 0.0015/0.49450 & 199.61 & 0.66/28.70 & 440.32 \\
\midrule
\multicolumn{9}{c}{Bounded-linear-rank-width circuits} \\
20 & 5 & 163.3 & 0.0050 & 5.04 & 0.0034/2.0102\phantom{(1)} & 200.68\phantom{(1)} & 1.39/\phantom{1}78.30 & 455.77 \\
20 & 7 & 160.0 & 0.0104 & 6.34 & 0.0034/2.2925\phantom{(1)}  & 201.47\phantom{(1)} & 1.31/\phantom{1}58.13 & 453.52 \\
30 & 5 & 324.8 & 0.0102 & 6.25 & 0.1457/6.0830\textcolor{red}{(1)} & 389.12\textcolor{red}{(1)} & 3.18/110.55 & 647.22 \\
30 & 7 & 321.2 & 0.0317 & 10.99 & 0.1387/5.7132\phantom{(1)}  & 330.47\phantom{(1)} & 3.24/110.97 & 703.99 \\
40 & 5 & 537.0 & 0.0146 & 7.50 & 4.8298/50.317\phantom{(1)}& 3960.82\phantom{(10)} & 140.87/268.17 & \phantom{1}66586.96 \\
40 & 7 & 532.4 & 0.0511 & 15.79 & 14.951/258.46\phantom{(1)} & 10794.56\phantom{(100)} & 186.46/312.41 & 102544.06 \\
\midrule
\multicolumn{9}{c}{Linear-network circuits} \\
20 & 5 & 162.0 & 0.002 & 12.88 & 0.0033/1.7682 & 202.93 & 1.81/34.78 & 473.47 \\
20 & 7 & 163.0 & 0.003 & 12.88 & 0.0031/2.1038 & 204.79 & 1.83/42.73 & 477.11 \\
30 & 5 & 319.8 & 0.003 & 12.90 & 0.0090/3.8352 & 213.92 & 3.31/86.50 & 539.16 \\
30 & 7 & 321.1 & 0.003 & 12.83 & 0.0154/4.2740 & 220.81 & 3.42/95.95 & 558.22 \\
40 & 5 & 527.3 & 0.004 & 12.87 & 0.7052/7.9125 & 932.70 & 5.80/127.72 & 821.91 \\
40 & 7 & 531.7 & 0.005 & 12.88 & 0.9697/11.2940 & 1433.34 & 6.20/129.52 & 1293.28 \\
\bottomrule
\end{tabular*}
\end{table}

\begin{table}[htbp]
    \centering
    \caption{Quantum circuit simulation benchmarks on Google circuits, using quimb with the
    \texttt{greedy} contraction-path optimizer. The notation, units, and metric definitions are the same as
those in Table~\ref{tab:lrw_results}.
    }
    \label{tab:google-greedy}
    \centering\footnotesize\setlength{\tabcolsep}{3pt}
    \begin{tabular}{cccc}
        \toprule
        Circuit family & $n$ & Time & Mem \\
        \midrule
       Google CZ v2 d10 & 16 & 0.0012/0.2816 & 217.74 \\
                            & 20 & 0.0013/0.2917 & 218.03 \\
                            & 25 & 0.0016/0.3070 & 219.14 \\
        \midrule
        Google iS v1 d10 & 16 & 0.0012/0.2675  & 216.52 \\
                            & 20 & 0.0011/0.2925 & 217.09 \\
                            & 25 & 0.0018/0.2922 & 217.19 \\
        \bottomrule
    \end{tabular}
\end{table}

\begin{figure}[htbp]
  \centering
  \includegraphics[width=\linewidth]{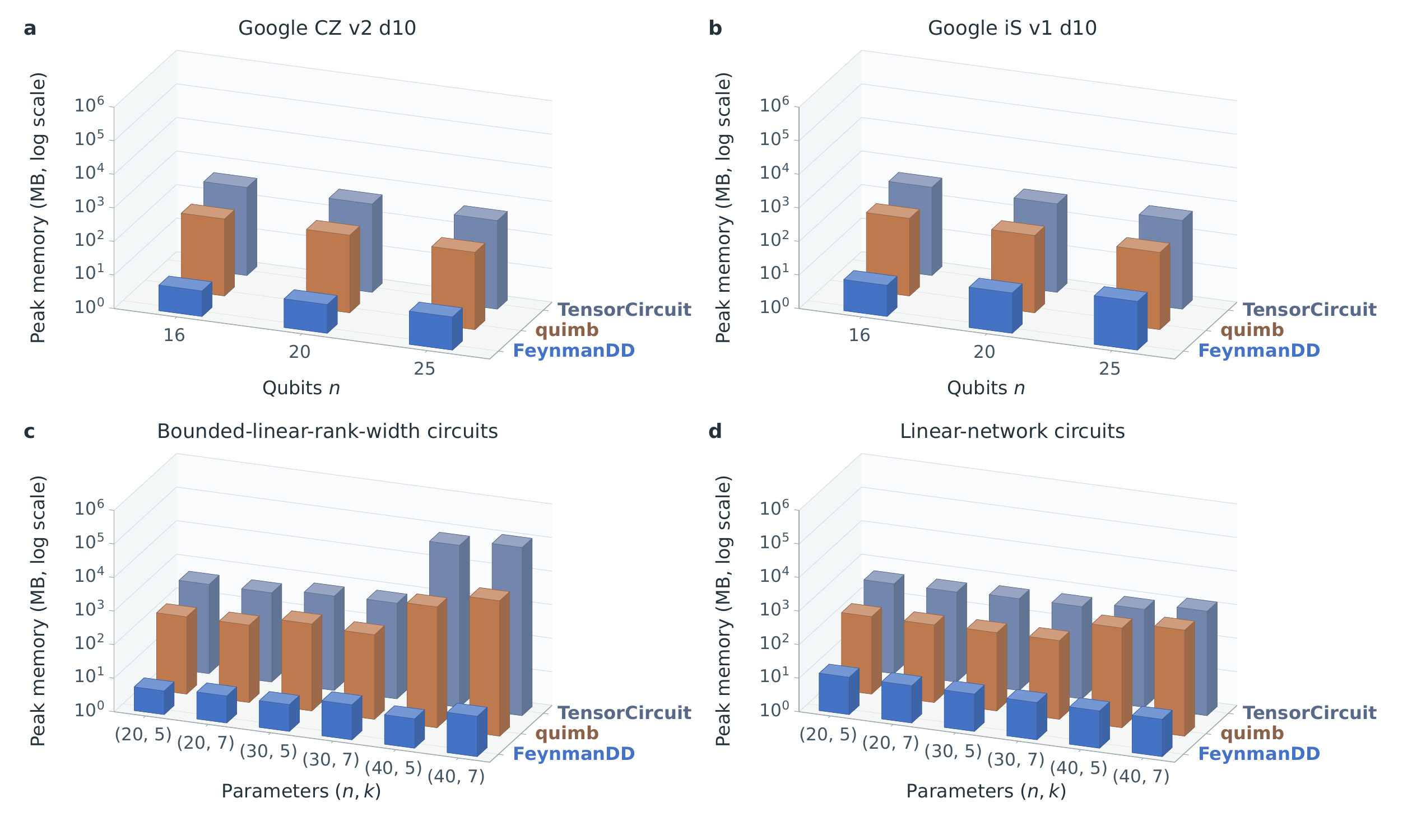}
  \caption{Peak memory usage of FeynmanDD, quimb, and TensorCircuit.
    The panels show (a) Google CZ v2 d10, (b) Google iS v1 d10, (c) bounded-linear-rank-width circuits, and (d) linear-network circuits.
    Bars show mean peak memory in MB on a logarithmic vertical scale.
    The numerical values and out-of-memory counts are given in \cref{tab:lrw_results}.}\label{fig:simulation-peak-memory}
\end{figure}

\FloatBarrier

%% file: main.bbl
\begin{thebibliography}{10}
\expandafter\ifx\csname url\endcsname\relax
  \def\url#1{\texttt{#1}}\fi
\expandafter\ifx\csname urlprefix\endcsname\relax\def\urlprefix{URL }\fi
\providecommand{\bibinfo}[2]{#2}
\providecommand{\eprint}[2][]{\url{#2}}

\bibitem{Got98}
\bibinfo{author}{Gottesman, D.}
\newblock \bibinfo{title}{The {Heisenberg} {Representation} of {Quantum}
  {Computers}}.
\newblock \emph{\bibinfo{journal}{arXiv:quant-ph/9807006}}
  (\bibinfo{year}{1998}).

\bibitem{Val02}
\bibinfo{author}{Valiant, L.~G.}
\newblock \bibinfo{title}{Quantum {Circuits} {That} {Can} {Be} {Simulated}
  {Classically} in {Polynomial} {Time}}.
\newblock \emph{\bibinfo{journal}{SIAM Journal on Computing}}
  \textbf{\bibinfo{volume}{31}}, \bibinfo{pages}{1229--1254}
  (\bibinfo{year}{2002}).

\bibitem{KLR+08}
\bibinfo{author}{Knill, E.} \emph{et~al.}
\newblock \bibinfo{title}{Randomized benchmarking of quantum gates}.
\newblock \emph{\bibinfo{journal}{Physical Review A}}
  \textbf{\bibinfo{volume}{77}}, \bibinfo{pages}{012307}
  (\bibinfo{year}{2008}).

\bibitem{MGE11}
\bibinfo{author}{Magesan, E.}, \bibinfo{author}{Gambetta, J.~M.} \&
  \bibinfo{author}{Emerson, J.}
\newblock \bibinfo{title}{Scalable and {Robust} {Randomized} {Benchmarking} of
  {Quantum} {Processes}}.
\newblock \emph{\bibinfo{journal}{Physical Review Letters}}
  \textbf{\bibinfo{volume}{106}}, \bibinfo{pages}{180504}
  (\bibinfo{year}{2011}).

\bibitem{AAB+19}
\bibinfo{author}{Arute, F.} \emph{et~al.}
\newblock \bibinfo{title}{Quantum supremacy using a programmable
  superconducting processor}.
\newblock \emph{\bibinfo{journal}{Nature}} \textbf{\bibinfo{volume}{574}},
  \bibinfo{pages}{505--510} (\bibinfo{year}{2019}).

\bibitem{ZWD+20}
\bibinfo{author}{Zhong, H.-S.} \emph{et~al.}
\newblock \bibinfo{title}{Quantum computational advantage using photons}.
\newblock \emph{\bibinfo{journal}{Science}} \textbf{\bibinfo{volume}{370}},
  \bibinfo{pages}{1460--1463} (\bibinfo{year}{2020}).

\bibitem{WBC+21}
\bibinfo{author}{Wu, Y.} \emph{et~al.}
\newblock \bibinfo{title}{Strong quantum computational advantage using a
  superconducting quantum processor}.
\newblock \emph{\bibinfo{journal}{Physical Review Letters}}
  \textbf{\bibinfo{volume}{127}}, \bibinfo{pages}{180501}
  (\bibinfo{year}{2021}).

\bibitem{MLA+22}
\bibinfo{author}{Madsen, L.~S.} \emph{et~al.}
\newblock \bibinfo{title}{Quantum computational advantage with a programmable
  photonic processor}.
\newblock \emph{\bibinfo{journal}{Nature}} \textbf{\bibinfo{volume}{606}},
  \bibinfo{pages}{75--81} (\bibinfo{year}{2022}).

\bibitem{GFZ+25}
\bibinfo{author}{Gao, D.} \emph{et~al.}
\newblock \bibinfo{title}{Establishing a new benchmark in quantum computational
  advantage with 105-qubit {Zuchongzhi 3.0} processor}.
\newblock \emph{\bibinfo{journal}{Physical Review Letters}}
  \textbf{\bibinfo{volume}{134}}, \bibinfo{pages}{090601}
  (\bibinfo{year}{2025}).

\bibitem{BG16}
\bibinfo{author}{Bravyi, S.} \& \bibinfo{author}{Gosset, D.}
\newblock \bibinfo{title}{Improved classical simulation of quantum circuits
  dominated by {Clifford} gates}.
\newblock \emph{\bibinfo{journal}{Physical Review Letters}}
  \textbf{\bibinfo{volume}{116}}, \bibinfo{pages}{250501}
  (\bibinfo{year}{2016}).

\bibitem{MS08}
\bibinfo{author}{Markov, I.~L.} \& \bibinfo{author}{Shi, Y.}
\newblock \bibinfo{title}{Simulating {Quantum} {Computation} by {Contracting}
  {Tensor} {Networks}}.
\newblock \emph{\bibinfo{journal}{SIAM Journal on Computing}}
  \textbf{\bibinfo{volume}{38}}, \bibinfo{pages}{963--981}
  (\bibinfo{year}{2008}).

\bibitem{HZN+21}
\bibinfo{author}{Huang, C.} \emph{et~al.}
\newblock \bibinfo{title}{Efficient parallelization of tensor network
  contraction for simulating quantum computation}.
\newblock \emph{\bibinfo{journal}{Nature Computational Science}}
  \textbf{\bibinfo{volume}{1}}, \bibinfo{pages}{578--587}
  (\bibinfo{year}{2021}).

\bibitem{PVWC07}
\bibinfo{author}{Perez-Garcia, D.}, \bibinfo{author}{Verstraete, F.},
  \bibinfo{author}{Wolf, M.~M.} \& \bibinfo{author}{Cirac, J.~I.}
\newblock \bibinfo{title}{Matrix {Product} {State} {Representations}}.
\newblock \emph{\bibinfo{journal}{arXiv:quant-ph/0608197}}
  (\bibinfo{year}{2007}).

\bibitem{OGo19}
\bibinfo{author}{O'Gorman, B.}
\newblock \bibinfo{title}{Parameterization of tensor network contraction}.
\newblock \emph{\bibinfo{journal}{LIPIcs, Volume 135, TQC 2019}}
  \textbf{\bibinfo{volume}{135}}, \bibinfo{pages}{10:1--10:19}
  (\bibinfo{year}{2019}).

\bibitem{Vid04}
\bibinfo{author}{Vidal, G.}
\newblock \bibinfo{title}{Efficient {Simulation} of {One}-{Dimensional}
  {Quantum} {Many}-{Body} {Systems}}.
\newblock \emph{\bibinfo{journal}{Physical Review Letters}}
  \textbf{\bibinfo{volume}{93}}, \bibinfo{pages}{040502}
  (\bibinfo{year}{2004}).

\bibitem{Vid08}
\bibinfo{author}{Vidal, G.}
\newblock \bibinfo{title}{Class of {Quantum} {Many}-{Body} {States} {That}
  {Can} {Be} {Efficiently} {Simulated}}.
\newblock \emph{\bibinfo{journal}{Physical Review Letters}}
  \textbf{\bibinfo{volume}{101}}, \bibinfo{pages}{110501}
  (\bibinfo{year}{2008}).

\bibitem{Oru14}
\bibinfo{author}{Or{\'u}s, R.}
\newblock \bibinfo{title}{A practical introduction to tensor networks: {Matrix}
  product states and projected entangled pair states}.
\newblock \emph{\bibinfo{journal}{Annals of Physics}}
  \textbf{\bibinfo{volume}{349}}, \bibinfo{pages}{117--158}
  (\bibinfo{year}{2014}).

\bibitem{ALM07}
\bibinfo{author}{Aharonov, D.}, \bibinfo{author}{Landau, Z.} \&
  \bibinfo{author}{Makowsky, J.}
\newblock \bibinfo{title}{The quantum {FFT} can be classically simulated}
  (\bibinfo{year}{2007}).

\bibitem{FG06}
\bibinfo{author}{Flum, J.} \& \bibinfo{author}{Grohe, M.}
\newblock \emph{\bibinfo{title}{Parameterized complexity theory}}.
\newblock Texts in theoretical computer science (\bibinfo{publisher}{Springer},
  \bibinfo{address}{Berlin ; New York}, \bibinfo{year}{2006}).

\bibitem{WCYJ25}
\bibinfo{author}{Wang, Z.}, \bibinfo{author}{Cheng, B.}, \bibinfo{author}{Yuan,
  L.} \& \bibinfo{author}{Ji, Z.}
\newblock \bibinfo{title}{{FeynmanDD: Quantum Circuit Analysis with Classical
  Decision Diagrams}}.
\newblock In \emph{\bibinfo{booktitle}{Proceedings of the 37th International
  Conference on Computer Aided Verification}} (\bibinfo{address}{Zagreb,
  Croatia}, \bibinfo{year}{2025}).

\bibitem{Bry86}
\bibinfo{author}{Bryant, R.~E.}
\newblock \bibinfo{title}{Graph-{Based} {Algorithms} for {Boolean} {Function}
  {Manipulation}}.
\newblock \emph{\bibinfo{journal}{IEEE Transactions on Computers}}
  \textbf{\bibinfo{volume}{C-35}}, \bibinfo{pages}{677--691}
  (\bibinfo{year}{1986}).

\bibitem{Bry95}
\bibinfo{author}{Bryant, R.}
\newblock \bibinfo{title}{Binary decision diagrams and beyond: enabling
  technologies for formal verification}.
\newblock In \emph{\bibinfo{booktitle}{Proceedings of {IEEE} {International}
  {Conference} on {Computer} {Aided} {Design} ({ICCAD})}},
  \bibinfo{pages}{236--243} (\bibinfo{publisher}{IEEE Comput. Soc. Press},
  \bibinfo{address}{San Jose, CA, USA}, \bibinfo{year}{1995}).

\bibitem{Weg00}
\bibinfo{author}{Wegener, I.}
\newblock \emph{\bibinfo{title}{Branching programs and binary decision
  diagrams: theory and applications}}.
\newblock {SIAM} monographs on discrete mathematics and applications
  (\bibinfo{publisher}{Society for Industrial and Applied Mathematics},
  \bibinfo{address}{Philadelphia}, \bibinfo{year}{2000}).

\bibitem{Knu09}
\bibinfo{author}{Knuth, D.~E.}
\newblock \emph{\bibinfo{title}{The {Art} of {Computer} {Programming}, {Volume}
  4, {Fascicle} 1 ({Bitwise} {Tricks} \& {Techniques}; {Binary} {Decision}
  {Diagrams})}} (\bibinfo{publisher}{AddisonWesley Professional},
  \bibinfo{address}{Upper Saddle River, NJ}, \bibinfo{year}{2009}),
  \bibinfo{edition}{1 edition} edn.

\bibitem{MMBS04}
\bibinfo{author}{Molitor, P.}, \bibinfo{author}{Mohnke, J.},
  \bibinfo{author}{Becker, B.} \& \bibinfo{author}{Scholl, C.}
\newblock \emph{\bibinfo{title}{Equivalence {Checking} of {Digital} {Circuits}:
  {Fundamentals}, {Principles}, {Methods}}} (\bibinfo{publisher}{Springer New
  York, NY}, \bibinfo{year}{2004}).

\bibitem{BFG+93}
\bibinfo{author}{Bahar, R.} \emph{et~al.}
\newblock \bibinfo{title}{Algebraic decision diagrams and their applications}.
\newblock In \emph{\bibinfo{booktitle}{Proceedings of 1993 {International}
  {Conference} on {Computer} {Aided} {Design} ({ICCAD})}},
  \bibinfo{pages}{188--191} (\bibinfo{publisher}{IEEE Comput. Soc. Press},
  \bibinfo{address}{Santa Clara, CA, USA}, \bibinfo{year}{1993}).

\bibitem{JKO17}
\bibinfo{author}{Jeong, J.}, \bibinfo{author}{Kim, E.~J.} \&
  \bibinfo{author}{Oum, S.-i.}
\newblock \bibinfo{title}{The “{Art} of {Trellis} {Decoding}” {Is}
  {Fixed}-{Parameter} {Tractable}}.
\newblock \emph{\bibinfo{journal}{IEEE Trans. Inform. Theory}}
  \textbf{\bibinfo{volume}{63}}, \bibinfo{pages}{7178--7205}
  (\bibinfo{year}{2017}).

\bibitem{JKO21}
\bibinfo{author}{Jeong, J.}, \bibinfo{author}{Kim, E.~J.} \&
  \bibinfo{author}{Oum, S.-i.}
\newblock \bibinfo{title}{Finding branch-decompositions of matroids,
  hypergraphs, and more}.
\newblock \emph{\bibinfo{journal}{SIAM Journal on Discrete Mathematics}}
  \textbf{\bibinfo{volume}{35}}, \bibinfo{pages}{2544--2617}
  (\bibinfo{year}{2021}).

\bibitem{KL23}
\bibinfo{author}{Korhonen, T.} \& \bibinfo{author}{Lokshtanov, D.}
\newblock \bibinfo{title}{An improved parameterized algorithm for treewidth}.
\newblock In \emph{\bibinfo{booktitle}{Proceedings of the 55th Annual ACM
  Symposium on Theory of Computing}}, \bibinfo{pages}{528--541}
  (\bibinfo{publisher}{Association for Computing Machinery},
  \bibinfo{year}{2023}).

\bibitem{ChD08}
\bibinfo{author}{Chavira, M.} \& \bibinfo{author}{Darwiche, A.}
\newblock \bibinfo{title}{On probabilistic inference by weighted model
  counting}.
\newblock \emph{\bibinfo{journal}{Artificial Intelligence}}
  \textbf{\bibinfo{volume}{172}}, \bibinfo{pages}{772--799}
  (\bibinfo{year}{2008}).

\bibitem{KVD17}
\bibinfo{author}{Kimmig, A.}, \bibinfo{author}{Van~den Broeck, G.} \&
  \bibinfo{author}{De~Raedt, L.}
\newblock \bibinfo{title}{Algebraic model counting}.
\newblock \emph{\bibinfo{journal}{Journal of Applied Logic}}
  \textbf{\bibinfo{volume}{22}}, \bibinfo{pages}{46--62}
  (\bibinfo{year}{2017}).

\bibitem{KFL01}
\bibinfo{author}{Kschischang, F.~R.}, \bibinfo{author}{Frey, B.~J.} \&
  \bibinfo{author}{Loeliger, H.-A.}
\newblock \bibinfo{title}{Factor graphs and the sum-product algorithm}.
\newblock \emph{\bibinfo{journal}{IEEE Transactions on Information Theory}}
  \textbf{\bibinfo{volume}{47}}, \bibinfo{pages}{498--519}
  (\bibinfo{year}{2001}).

\bibitem{NC00}
\bibinfo{author}{Nielsen, M.~A.} \& \bibinfo{author}{Chuang, I.~L.}
\newblock \emph{\bibinfo{title}{Quantum {Computation} and {Quantum}
  {Information}}} (\bibinfo{publisher}{Cambridge University Press},
  \bibinfo{year}{2000}).

\bibitem{dCGH+26}
\bibinfo{author}{de~Colnet, A.} \emph{et~al.}
\newblock \bibinfo{title}{Quadratic sums-of-powers for fixed-parameter
  tractable quantum-circuit simulation}.
\newblock Preprint at \url{https://arxiv.org/abs/2605.29944v2}
  (\bibinfo{year}{2026}).

\bibitem{AK15}
\bibinfo{author}{Adler, I.} \& \bibinfo{author}{Kant{\'e}, M.~M.}
\newblock \bibinfo{title}{Linear rank-width and linear clique-width of trees}.
\newblock \emph{\bibinfo{journal}{Theoretical Computer Science}}
  \textbf{\bibinfo{volume}{589}}, \bibinfo{pages}{87--98}
  (\bibinfo{year}{2015}).

\bibitem{KS93}
\bibinfo{author}{Korach, E.} \& \bibinfo{author}{Solel, N.}
\newblock \bibinfo{title}{Tree-width, path-width, and cutwidth}.
\newblock \emph{\bibinfo{journal}{Discrete Applied Mathematics}}
  \textbf{\bibinfo{volume}{43}}, \bibinfo{pages}{97--101}
  (\bibinfo{year}{1993}).

\bibitem{SB09}
\bibinfo{author}{Shepherd, D.} \& \bibinfo{author}{Bremner, M.~J.}
\newblock \bibinfo{title}{Temporally unstructured quantum computation}.
\newblock \emph{\bibinfo{journal}{Proc. R. Soc. A}}
  \textbf{\bibinfo{volume}{465}}, \bibinfo{pages}{1413--1439}
  (\bibinfo{year}{2009}).

\bibitem{BJS11}
\bibinfo{author}{Bremner, M.~J.}, \bibinfo{author}{Jozsa, R.} \&
  \bibinfo{author}{Shepherd, D.~J.}
\newblock \bibinfo{title}{Classical simulation of commuting quantum
  computations implies collapse of the polynomial hierarchy}.
\newblock \emph{\bibinfo{journal}{Proc. R. Soc. A}}
  \textbf{\bibinfo{volume}{467}}, \bibinfo{pages}{459--472}
  (\bibinfo{year}{2011}).

\bibitem{BMS16}
\bibinfo{author}{Bremner, M.~J.}, \bibinfo{author}{Montanaro, A.} \&
  \bibinfo{author}{Shepherd, D.~J.}
\newblock \bibinfo{title}{Average-{Case} {Complexity} {Versus} {Approximate}
  {Simulation} of {Commuting} {Quantum} {Computations}}.
\newblock \emph{\bibinfo{journal}{Phys. Rev. Lett.}}
  \textbf{\bibinfo{volume}{117}}, \bibinfo{pages}{080501}
  (\bibinfo{year}{2016}).

\bibitem{JW12}
\bibinfo{author}{Joret, G.} \& \bibinfo{author}{Wood, D.~R.}
\newblock \bibinfo{title}{Nordhaus-{Gaddum} for treewidth}.
\newblock \emph{\bibinfo{journal}{European Journal of Combinatorics}}
  \textbf{\bibinfo{volume}{33}}, \bibinfo{pages}{488--490}
  (\bibinfo{year}{2012}).

\bibitem{Gra18}
\bibinfo{author}{Gray, J.}
\newblock \bibinfo{title}{quimb: a python library for quantum information and
  many-body calculations}.
\newblock \emph{\bibinfo{journal}{Journal of Open Source Software}}
  \textbf{\bibinfo{volume}{3}}, \bibinfo{pages}{819} (\bibinfo{year}{2018}).

\bibitem{ZAW+23}
\bibinfo{author}{Zhang, S.-X.} \emph{et~al.}
\newblock \bibinfo{title}{{TensorCircuit}: a quantum software framework for the
  {NISQ} era}.
\newblock \emph{\bibinfo{journal}{Quantum}} \textbf{\bibinfo{volume}{7}},
  \bibinfo{pages}{912} (\bibinfo{year}{2023}).

\bibitem{ZW19}
\bibinfo{author}{Zulehner, A.} \& \bibinfo{author}{Wille, R.}
\newblock \bibinfo{title}{Advanced {Simulation} of {Quantum} {Computations}}.
\newblock \emph{\bibinfo{journal}{IEEE Transactions on Computer-Aided Design of
  Integrated Circuits and Systems}} \textbf{\bibinfo{volume}{38}},
  \bibinfo{pages}{848--859} (\bibinfo{year}{2019}).

\bibitem{ZHW19}
\bibinfo{author}{Zulehner, A.}, \bibinfo{author}{Hillmich, S.} \&
  \bibinfo{author}{Wille, R.}
\newblock \bibinfo{title}{How to {Efficiently} {Handle} {Complex} {Values}?
  {Implementing} {Decision} {Diagrams} for {Quantum} {Computing}}.
\newblock \emph{\bibinfo{journal}{arXiv:1911.12691 [quant-ph]}}
  (\bibinfo{year}{2019}).

\bibitem{SCR23}
\bibinfo{author}{Sistla, M.}, \bibinfo{author}{Chaudhuri, S.} \&
  \bibinfo{author}{Reps, T.}
\newblock \bibinfo{title}{Symbolic quantum simulation with quasimodo}.
\newblock In \emph{\bibinfo{booktitle}{Computer Aided Verification: 35th
  International Conference, CAV 2023, Paris, France, July 17--22, 2023,
  Proceedings, Part III}}, \bibinfo{pages}{213--225}
  (\bibinfo{publisher}{Springer-Verlag}, \bibinfo{address}{Berlin, Heidelberg},
  \bibinfo{year}{2023}).

\bibitem{SCR24}
\bibinfo{author}{Sistla, M.}, \bibinfo{author}{Chaudhuri, S.} \&
  \bibinfo{author}{Reps, T.}
\newblock \bibinfo{title}{Weighted {Context}-{Free}-{Language} {Ordered}
  {Binary} {Decision} {Diagrams}}.
\newblock \emph{\bibinfo{journal}{Weighted CFLOBDDs}}
  \textbf{\bibinfo{volume}{8}}, \bibinfo{pages}{320:1390--320:1419}
  (\bibinfo{year}{2024}).

\bibitem{HEB04}
\bibinfo{author}{Hein, M.}, \bibinfo{author}{Eisert, J.} \&
  \bibinfo{author}{Briegel, H.~J.}
\newblock \bibinfo{title}{Multiparty entanglement in graph states}.
\newblock \emph{\bibinfo{journal}{Phys. Rev. A}} \textbf{\bibinfo{volume}{69}},
  \bibinfo{pages}{062311} (\bibinfo{year}{2004}).

\bibitem{DW18}
\bibinfo{author}{Dahlberg, A.} \& \bibinfo{author}{Wehner, S.}
\newblock \bibinfo{title}{Transforming graph states using single-qubit
  operations}.
\newblock \emph{\bibinfo{journal}{Phil. Trans. R. Soc. A.}}
  \textbf{\bibinfo{volume}{376}}, \bibinfo{pages}{20170325}
  (\bibinfo{year}{2018}).

\bibitem{AMDS20}
\bibinfo{author}{Adcock, J.~C.}, \bibinfo{author}{Morley-Short, S.},
  \bibinfo{author}{Dahlberg, A.} \& \bibinfo{author}{Silverstone, J.~W.}
\newblock \bibinfo{title}{Mapping graph state orbits under local
  complementation}.
\newblock \emph{\bibinfo{journal}{Quantum}} \textbf{\bibinfo{volume}{4}},
  \bibinfo{pages}{305} (\bibinfo{year}{2020}).

\bibitem{V-d-NDVB07}
\bibinfo{author}{Van~den Nest, M.}, \bibinfo{author}{D{\"u}r, W.},
  \bibinfo{author}{Vidal, G.} \& \bibinfo{author}{Briegel, H.~J.}
\newblock \bibinfo{title}{Classical simulation versus universality in
  measurement-based quantum computation}.
\newblock \emph{\bibinfo{journal}{Physical Review A}}
  \textbf{\bibinfo{volume}{75}}, \bibinfo{pages}{012337}
  (\bibinfo{year}{2007}).

\bibitem{HI14}
\bibinfo{author}{Hiraishi, H.} \& \bibinfo{author}{Imai, H.}
\newblock \bibinfo{title}{{BDD} {Operations} for {Quantum} {Graph} {States}}.
\newblock In \bibinfo{editor}{Yamashita, S.} \& \bibinfo{editor}{Minato, S.-i.}
  (eds.) \emph{\bibinfo{booktitle}{Reversible {Computation}}},
  \bibinfo{pages}{216--229} (\bibinfo{publisher}{Springer International
  Publishing}, \bibinfo{address}{Cham}, \bibinfo{year}{2014}).

\bibitem{HIIL18}
\bibinfo{author}{Hiraishi, H.}, \bibinfo{author}{Imai, H.},
  \bibinfo{author}{Iwata, Y.} \& \bibinfo{author}{Lin, B.}
\newblock \bibinfo{title}{Parameterized algorithms to compute {Ising} partition
  function}.
\newblock \emph{\bibinfo{journal}{IEICE Transactions on Fundamentals of
  Electronics, Communications and Computer Sciences}}
  \textbf{\bibinfo{volume}{E101.A}}, \bibinfo{pages}{1398--1403}
  (\bibinfo{year}{2018}).

\bibitem{KK26}
\bibinfo{author}{Kuyanov, F.} \& \bibinfo{author}{Kissinger, A.}
\newblock \bibinfo{title}{Efficient classical simulation of low-rank-width
  quantum circuits using {ZX}-calculus}.
\newblock Preprint at \url{https://arxiv.org/abs/2603.06764}
  (\bibinfo{year}{2026}).

\bibitem{CL26}
\bibinfo{author}{Codsi, J.} \& \bibinfo{author}{Laakkonen, T.}
\newblock \bibinfo{title}{Unifying graph measures and stabilizer decompositions
  for the classical simulation of quantum circuits}.
\newblock Preprint at \url{https://arxiv.org/abs/2603.06377}
  (\bibinfo{year}{2026}).

\bibitem{MBL24}
\bibinfo{author}{Mei, J.}, \bibinfo{author}{Bonsangue, M.} \&
  \bibinfo{author}{Laarman, A.}
\newblock \bibinfo{title}{Simulating {Quantum} {Circuits} by {Model}
  {Counting}} (\bibinfo{year}{2024}).

\bibitem{vdELLB25}
\bibinfo{author}{van~den Ende, D.}, \bibinfo{author}{Lee, J.~H.},
  \bibinfo{author}{Laarman, A.} \& \bibinfo{author}{Basold, H.}
\newblock \bibinfo{title}{Quantum physics using weighted model counting}.
\newblock Preprint at \url{https://arxiv.org/abs/2508.21288v2}
  (\bibinfo{year}{2025}).

\bibitem{HCC+26}
\bibinfo{author}{Huang, W.-J.} \emph{et~al.}
\newblock \bibinfo{title}{Equivalence checking of quantum circuits via path-sum
  and weighted model counting}.
\newblock In \emph{\bibinfo{booktitle}{Tools and Algorithms for the
  Construction and Analysis of Systems}}, vol. \bibinfo{volume}{16506} of
  \emph{\bibinfo{series}{Lecture Notes in Computer Science}},
  \bibinfo{pages}{419--439} (\bibinfo{publisher}{Springer},
  \bibinfo{year}{2026}).
\newblock arXiv:2604.24504.

\bibitem{WJ08}
\bibinfo{author}{Wainwright, M.~J.} \& \bibinfo{author}{Jordan, M.~I.}
\newblock \bibinfo{title}{Graphical models, exponential families, and
  variational inference}.
\newblock \emph{\bibinfo{journal}{Foundations and Trends in Machine Learning}}
  \textbf{\bibinfo{volume}{1}}, \bibinfo{pages}{1--305} (\bibinfo{year}{2008}).

\bibitem{DKLP02}
\bibinfo{author}{Dennis, E.}, \bibinfo{author}{Kitaev, A.},
  \bibinfo{author}{Landahl, A.} \& \bibinfo{author}{Preskill, J.}
\newblock \bibinfo{title}{Topological quantum memory}.
\newblock \emph{\bibinfo{journal}{Journal of Mathematical Physics}}
  \textbf{\bibinfo{volume}{43}}, \bibinfo{pages}{4452--4505}
  (\bibinfo{year}{2002}).

\bibitem{CP26}
\bibinfo{author}{Cheng, B.} \& \bibinfo{author}{Pan, F.}
\newblock \bibinfo{title}{Exact maximum likelihood decoding beyond treewidth
  via rank-decomposition dynamic programming}.
\newblock Preprint at \url{https://arxiv.org/abs/2609.39556}
  (\bibinfo{year}{2026}).

\bibitem{OS06}
\bibinfo{author}{Oum, S.-i.} \& \bibinfo{author}{Seymour, P.}
\newblock \bibinfo{title}{Approximating clique-width and branch-width}.
\newblock \emph{\bibinfo{journal}{Journal of Combinatorial Theory, Series B}}
  \textbf{\bibinfo{volume}{96}}, \bibinfo{pages}{514--528}
  (\bibinfo{year}{2006}).

\bibitem{BISN18}
\bibinfo{author}{Boixo, S.}, \bibinfo{author}{Isakov, S.~V.},
  \bibinfo{author}{Smelyanskiy, V.~N.} \& \bibinfo{author}{Neven, H.}
\newblock \bibinfo{title}{Simulation of low-depth quantum circuits as complex
  undirected graphical models} (\bibinfo{year}{2018}).

\end{thebibliography}

\begin{thebibliography}{10}
\expandafter\ifx\csname url\endcsname\relax
  \def\url#1{\texttt{#1}}\fi
\expandafter\ifx\csname urlprefix\endcsname\relax\def\urlprefix{URL }\fi
\providecommand{\bibinfo}[2]{#2}
\providecommand{\eprint}[2][]{\url{#2}}

\bibitem{Knu09}
\bibinfo{author}{Knuth, D.~E.}
\newblock \emph{\bibinfo{title}{The {Art} of {Computer} {Programming}, {Volume}
  4, {Fascicle} 1 ({Bitwise} {Tricks} \& {Techniques}; {Binary} {Decision}
  {Diagrams})}} (\bibinfo{publisher}{AddisonWesley Professional},
  \bibinfo{address}{Upper Saddle River, NJ}, \bibinfo{year}{2009}),
  \bibinfo{edition}{1 edition} edn.

\bibitem{AAB+19}
\bibinfo{author}{Arute, F.} \emph{et~al.}
\newblock \bibinfo{title}{Quantum supremacy using a programmable
  superconducting processor}.
\newblock \emph{\bibinfo{journal}{Nature}} \textbf{\bibinfo{volume}{574}},
  \bibinfo{pages}{505--510} (\bibinfo{year}{2019}).

\bibitem{WCYJ25}
\bibinfo{author}{Wang, Z.}, \bibinfo{author}{Cheng, B.}, \bibinfo{author}{Yuan,
  L.} \& \bibinfo{author}{Ji, Z.}
\newblock \bibinfo{title}{{FeynmanDD: Quantum Circuit Analysis with Classical
  Decision Diagrams}}.
\newblock In \emph{\bibinfo{booktitle}{Proceedings of the 37th International
  Conference on Computer Aided Verification}} (\bibinfo{address}{Zagreb,
  Croatia}, \bibinfo{year}{2025}).

\bibitem{JKO17}
\bibinfo{author}{Jeong, J.}, \bibinfo{author}{Kim, E.~J.} \&
  \bibinfo{author}{Oum, S.-i.}
\newblock \bibinfo{title}{The “{Art} of {Trellis} {Decoding}” {Is}
  {Fixed}-{Parameter} {Tractable}}.
\newblock \emph{\bibinfo{journal}{IEEE Trans. Inform. Theory}}
  \textbf{\bibinfo{volume}{63}}, \bibinfo{pages}{7178--7205}
  (\bibinfo{year}{2017}).

\bibitem{JKO21}
\bibinfo{author}{Jeong, J.}, \bibinfo{author}{Kim, E.~J.} \&
  \bibinfo{author}{Oum, S.-i.}
\newblock \bibinfo{title}{Finding branch-decompositions of matroids,
  hypergraphs, and more}.
\newblock \emph{\bibinfo{journal}{SIAM Journal on Discrete Mathematics}}
  \textbf{\bibinfo{volume}{35}}, \bibinfo{pages}{2544--2617}
  (\bibinfo{year}{2021}).

\bibitem{Bry86}
\bibinfo{author}{Bryant, R.~E.}
\newblock \bibinfo{title}{Graph-{Based} {Algorithms} for {Boolean} {Function}
  {Manipulation}}.
\newblock \emph{\bibinfo{journal}{IEEE Transactions on Computers}}
  \textbf{\bibinfo{volume}{C-35}}, \bibinfo{pages}{677--691}
  (\bibinfo{year}{1986}).

\bibitem{KVD17}
\bibinfo{author}{Kimmig, A.}, \bibinfo{author}{Van~den Broeck, G.} \&
  \bibinfo{author}{De~Raedt, L.}
\newblock \bibinfo{title}{Algebraic model counting}.
\newblock \emph{\bibinfo{journal}{Journal of Applied Logic}}
  \textbf{\bibinfo{volume}{22}}, \bibinfo{pages}{46--62}
  (\bibinfo{year}{2017}).

\bibitem{NC00}
\bibinfo{author}{Nielsen, M.~A.} \& \bibinfo{author}{Chuang, I.~L.}
\newblock \emph{\bibinfo{title}{Quantum {Computation} and {Quantum}
  {Information}}} (\bibinfo{publisher}{Cambridge University Press},
  \bibinfo{year}{2000}).

\bibitem{MS08}
\bibinfo{author}{Markov, I.~L.} \& \bibinfo{author}{Shi, Y.}
\newblock \bibinfo{title}{Simulating {Quantum} {Computation} by {Contracting}
  {Tensor} {Networks}}.
\newblock \emph{\bibinfo{journal}{SIAM Journal on Computing}}
  \textbf{\bibinfo{volume}{38}}, \bibinfo{pages}{963--981}
  (\bibinfo{year}{2008}).

\bibitem{JW12}
\bibinfo{author}{Joret, G.} \& \bibinfo{author}{Wood, D.~R.}
\newblock \bibinfo{title}{Nordhaus-{Gaddum} for treewidth}.
\newblock \emph{\bibinfo{journal}{European Journal of Combinatorics}}
  \textbf{\bibinfo{volume}{33}}, \bibinfo{pages}{488--490}
  (\bibinfo{year}{2012}).

\end{thebibliography}
